\documentclass[12pt]{article}
\usepackage[english]{babel}
\usepackage{amsthm,amssymb,amsfonts,amsmath,amscd}
\usepackage{graphicx}
\usepackage{nccmath}

\makeatletter \@addtoreset{equation}{section} \makeatother
\newtheorem{proposition}{Proposition}
\newtheorem{theorem}{Theorem}
\newtheorem{lemma}{Lemma}
\newtheorem{remark}{Remark}

\newcommand{\mdet}{\mathrm{det}}
\newcommand{\intd}{\displaystyle\int}
\newcommand{\Tr}{\mathrm{Tr}\,}
\begin{document}

\title{Universality of the local regime for the block band matrices with a finite number of blocks}
%\titlerunning{Universality  for the block band matrices}

\author{ Tatyana Shcherbina
\thanks{Chebyshev Laboratory, Mathematical Department of St. Petersburg State University}
\thanks{IAS, Princeton, USA}}
%\\
%\small{IAS, Princeton, USA}}}
%  14th Line, 29b, Saint Petersburg, 199178 Russia\\
%                 \email{t\underline{ }shcherbina@rambler.ru}\\
%\and
%T. Shcherbina\at
%Institute for Advanced Study, Einstein Drive, Princeton, NJ 08540 USA
%}
\date{}
% The correct dates will be entered by the editor

\maketitle

\begin{abstract}
We consider the block band matrices, i.e. the Hermitian matrices $H_N$, $N=|\Lambda|W$ with elements $H_{jk,\alpha\beta}$,
where $j,k \in\Lambda=[1,m]^d\cap \mathbb{Z}^d$ (they parameterize the lattice sites) and $\alpha, \beta= 1,\ldots, W$ (they
parameterize the orbitals on each site). The entries $H_{jk,\alpha\beta}$ are random
Gaussian variables with mean zero such that $\langle H_{j_1k_1,\alpha_1\beta_1}H_{j_2k_2,\alpha_2\beta_2}\rangle=\delta_{j_1k_2}\delta_{j_2k_1}
\delta_{\alpha_1\beta_2}\delta_{\beta_1\alpha_2} J_{j_1k_1},$ where $J=1/W+\alpha\Delta/W$, $\alpha< 1/4d$.
This matrices are the special case of Wegner's $W$-orbital models. Assuming that the number of sites $|\Lambda|$ is finite,
we prove universality of the local eigenvalue statistics of $H_N$ for the energies $|\lambda_0|< \sqrt{2}$.
%\keywords{random matrices \and Wegner model\and band matrices \and universality}
%%% \PACS{PACS code1 \and PACS code2 \and more}
%% \subclass{MSC 15A52 \and MSC 15A57}
\end{abstract}
\section{Introduction}
Let $\Lambda=[1,m]^d\cap \mathbb{Z}^d$ be a periodic
box with volume $|\Lambda|=m^d$. Assign to every site $j\in\Lambda$ one copy $K_j\simeq \mathbb{C}^W$ of an $W$-dimensional complex
vector space, and set
\[
K=\bigoplus\limits_{j\in\Lambda}K_j\simeq\mathbb{C}^{|\Lambda|W}.
\]
From the physical point of view, we are assigning $W$ valence electron orbitals to
every atom of a solid with hypercubic lattice structure.

We start from the matrices $H: K\to K$ belonging to the Gaussian Unitary Ensemble (GUE), i.e.
from the Hermitian matrices with i.i.d. (modulo symmetry) Gaussian entries with mean zero
and variance 1, and then multiply the variances of all
matrix elements of $H$ connecting $K_j$ and $K_k$ by the positive number $J_{jk}$, $j,k\in\Lambda$ (which means
that $H$ becomes the matrix constructed of $W\times W$ blocks, and the variance in each block is constant).
%From the physical point of view, each atom $j\in \Lambda$ here is viewed as a grain,
%or small metallic particle, housing
%a large number $W$ of electron states, and the squared matrix elements for tunneling
%between grains ($J_{jk}$ for $j\ne k$) are smaller than the intra-grain matrix elements
%($J_{jj}$):
%\[
%J_{jj}>\sum\limits_{k\ne j} J_{jk}.
%\]
%Such models are called ``the granular''.

More precisely, we consider Hermitian matrices $H_N$, $N=|\Lambda|W$ with elements $H_{jk,\alpha\beta}$,
where $j,k \in\Lambda$ (they parameterize the lattice sites) and $\alpha, \beta= 1,\ldots, W$ (they
parameterize the orbitals on each site). The entries $H_{jk,\alpha\beta}$ are random
Gaussian variables with mean zero such that
\begin{equation}\label{H}
\langle H_{j_1k_1,\alpha_1\beta_1}H_{j_2k_2,\alpha_2\beta_2}\rangle=\delta_{j_1k_2}\delta_{j_2k_1}
\delta_{\alpha_1\beta_2}\delta_{\beta_1\alpha_2} J_{j_1k_1}.
\end{equation}
Here $J_{jk}\ge 0$ are matrix elements of the positive-definite symmetric $|\Lambda|\times |\Lambda|$ matrix $J$, such that
\[
\sum\limits_{j\in\Lambda}J_{jk}=1/W.
\]
The probability law of $H_N$ can be written in the form
\begin{equation}\label{pr_l}
P_N(d H_N)=\exp\Big\{-\dfrac{1}{2}\sum\limits_{j,k\in\Lambda}\sum\limits_{\alpha,\beta=1}^W
\dfrac{|H_{jk,\alpha\beta}|^2}{J_{jk}}\Big\}dH_N,
\end{equation}
where
\[
dH_N=\prod\limits_{j<k}\prod\limits_{\alpha\beta}\dfrac{dH_{jk,\alpha\beta}d\overline{H}_{jk,\alpha\beta}}{2\pi J_{jk}}
\prod\limits_{j}\prod\limits_{\alpha<\beta}\dfrac{dH_{jj,\alpha\beta}d\overline{H}_{jj,\alpha\beta}}{2\pi J_{jj}}
\prod\limits_{j}\prod\limits_{\alpha}\dfrac{dH_{jj,\alpha\alpha}}{\sqrt{2\pi J_{jj}}}.
\]
Such models were first introduced and studied by Wegner (see \cite{S-W:80}, \cite{We:79}).

Note that $P_N(d H_N)$ is invariant under conjugation
$H_N\to U^*H_NU$ by $U\in \mathcal{U}$, where $\mathcal{U}$ is the direct product of all the groups $U(K_j)$ of
unitary transformations in the subspaces:
$$\mathcal{U} = \bigotimes\limits_{j\in\Lambda}U(K_j).$$
This means that the probability distribution $P_N(d H_N)$ has a {\it local gauge
invariance}.

Varying the lattice $\Lambda$, the number of orbitals $W$, and the variances $J_{jk}$, one obtains
a large class of Hermitian random matrix ensembles. For example, putting $|\Lambda|=1$ we
get a zero-dimensional $W$-orbital model which coincides with GUE. From the other hand,
choice $J_{jk}=\varphi(|j-k|)$, where $\varphi$ is a rapidly decreasing positive function, gives
an ensemble of random band matrices.

Here we consider
\begin{equation}\label{J}
J=1/W+\alpha\Delta/W, \quad \alpha<1/4d,
\end{equation}
where $W\gg 1$ and $\Delta$ is the discrete Laplacian on $\Lambda$ with periodic boundary conditions.
This model is one of the possible realizations of the Gaussian random band matrices, for example for
$d=1$ they correspond to the band matrices with the width of the band $2W+1$.

Random band matrices are natural interpolations between random
Schr\"o\-din\-ger matrices $ H_{RS}=-\Delta+\lambda V$, in which the randomness
only appears in the diagonal potential $V$ ($\lambda$ is a small parameter which
measures the strength of the disorder) and
mean-field random matrices such as $N\times N$ Wigner matrices, i.e. Hermitian
random matrices with i.i.d elements. Moreover, random Schr$\ddot{\hbox{o}}$dinger matrices with parameter
$\lambda$ and RBM with the width of the band $W$ are expected to have some similar qualitative
properties when $\lambda\approx W^{-1}$
(for more details on these conjectures see \cite{Sp:12}).

The key physical parameter of these models is the localization length,
which describes the typical length scale of the eigenvectors of random matrices.
The system is called delocalized if the localization length $\ell$ is
comparable with the matrix size, and it is called localized otherwise.
Delocalized systems correspond to electric
conductors, and localized systems are insulators.

In the case of 1D RBM there is a physical conjecture (see \cite{Ca-Co:90}, \cite{FM:91}) stating that $\ell$ is
of order $W^2$ (for the energy in the bulk of the spectrum), which means that varying $W$
we can see the crossover: for $W\gg \sqrt{N}$ the eigenvectors are expected to be
delocalized and for $W\ll \sqrt{N}$ they are localized. In terms of eigenvalues
this means that the local eigenvalue statistics in the bulk of the spectrum changes from Poisson,
for $W\ll \sqrt{N}$, to GUE
(Hermitian matrices with i.i.d Gaussian elements),
for $W\gg \sqrt{N}$. For $d=2$ the localization length is expected to be exponentially
growing in $W$ (and so the critical value is $N\sim \log W$), and $\ell\sim N$ for $d \ge 3$, i.e.
the system is delocalized.
At the present time only some upper and lower
bounds for $\ell$ are proven rigorously. It is known from the paper \cite{S:09} that $\ell\le W^8$
for $d=1$. On the other side, in the resent papers \cite{EK:11}, \cite{Yau:12} it was proven first that $\ell\gg W^{7/6}$,
and then that $\ell\gg W^{5/4}$.

The questions of the order of the localization length are closely related to the universality conjecture
of the bulk local regime of the random matrix theory, which we briefly outline now.

Let $\lambda_1^{(N)},\ldots,\lambda_N^{(N)}$ be the eigenvalues of
$H_N$. Define their Normalized Co\-un\-ting Measure
(NCM) as
\begin{equation} \label{NCM}
\mathcal{N}_N(\sigma)=\sharp\{\lambda_j^{(N)}\in
\sigma,j=1,\ldots,N \}/N,\quad \mathcal{N}_N(\mathbb{R})=1,
\end{equation}
where $\sigma$ is an arbitrary interval of the real axis.
The behavior of $\mathcal{N}_N$ as $N\to\infty$ was studied for many ensembles.
For 1D RBM it was shown
in \cite{BMP:91}, \cite{MPK:92} that $\mathcal{N}_{N}$ converges weakly, as $N,W\to\infty$, to a non-random measure
$\mathcal{\mathcal{N}}$, which is called the limiting NCM of the ensemble.
The measure $\mathcal{N}$ is absolutely continuous
and its density $\rho$ is given by the well-known Wigner semicircle law (the same result
is valid for Wigner ensembles, in particular, for Gaussian ensembles GUE, GOE):
\begin{equation}\label{rho}
\rho(\lambda)=\dfrac{1}{2\pi}\sqrt{4-\lambda^2},\quad \lambda\in[-2,2].
\end{equation}
The same is valid for the matrices (\ref{H}) -- (\ref{J}).

More delicate result about the density of states is proven in
\cite{DPS:02} for 3D RBM, and in \cite{Con:87} for some types of Wegner models.

These results characterize the so-called global distribution of the eigenvalues.

The local regime deals with the behavior of eigenvalues of $N\times N$
random matrices on the intervals whose length is of the order of the mean distance between
nearest eigenvalues. The main objects of the local regime are $k$-point correlation functions
$R_k$ ($k=1,2,\ldots$), which can be defined by the equalities:
\begin{multline} \label{R}
\mathbf{E}\left\{ \sum_{j_{1}\neq ...\neq j_{k}}\varphi_k
(\lambda_{j_{1}}^{(N)},\dots,\lambda_{j_{k}}^{(N)})\right\}\\ =\int_{\mathbb{R}^{k}} \varphi_{k}
(\lambda_{1}^{(N)},\ldots,\lambda_{k}^{(N)})R_{k}(\lambda_{1}^{(N)},\ldots,\lambda_{k}^{(N)})
d\lambda_{1}^{(N)}\ldots d\lambda_{k}^{(N)},
\end{multline}
where $\varphi_{k}: \mathbb{R}^{k}\rightarrow \mathbb{C}$ is
bounded, continuous and symmetric in its arguments and the
summation is over all $k$-tuples of distinct integers $
j_{1},\dots,j_{k}\in\{1,\ldots,N\}$.

According to the Wigner -- Dyson universality conjecture (see e.g. \cite{Me:91}), the local behavior
of the eigenvalues does not depend on the matrix probability
law (ensemble) and is determined only by the symmetry type of matrices (real
symmetric, Hermitian, or quaternion real in the case of real eigenvalues and orthogonal,
unitary or symplectic in the case of eigenvalues on the unit circle).
For example, the conjecture states that for Hermitian random matrices in the bulk of the spectrum
and in the range of parameters for which the eigenvectors are
delocalized
\begin{multline}\label{Un}
\displaystyle\frac{1}{(N\rho(\lambda_0))^k}
R_k\left(\lambda_0+\displaystyle\frac{\xi_1}{\rho(\lambda_0)\,N},
\ldots,\lambda_0+\displaystyle\frac{\xi_k}{\rho(\lambda_0)\,N}\right)\\
\stackrel{w}{\longrightarrow}\det \Big\{\dfrac{\sin \pi(\xi_i-\xi_j)}
{\pi(\xi_i-\xi_j)}\Big\}_{i,j=1}^k,\quad N\to\infty
\end{multline}
for any fixed $k$. This means that the limit coincides with that for GUE.

In this language the conjecture about the crossover for RBM states that we get (\ref{Un})
for $W$, which correspond to delocalized states, and
we get another behavior, which is determined by the Poisson statistics, for $W$
which correspond to localized states. For the general Hermitian Wigner matrices (i.e. $|\Lambda|=1$, but the
distribution of matrix elements are not necessary Gaussian) bulk universality has been proved recently
in \cite{EYY:10}, \cite{TV:11}. However, in the general case
of RBM the question of bulk universality of local spectral statistics
is still open even for 1D Gaussian RBM.

In this paper we prove (\ref{Un}) for the second correlation function $R_2$ of the ensemble
(\ref{H}) -- (\ref{J}), if $W\to\infty$, but the number of sites is fixed (i.e. $m$ is finite).

An additional source of motivation for the current work is the development of the supersymmetric method
(SUSY) in the context of random operators with non-trivial spatial structures. This method is widely used in
the physics literature (see e.g. \cite{Ef},\cite{M:00}) and is potentially very powerful but the rigorous
control of the
integral representations, which can be obtained by this method, is difficult and so far for the band
matrices (and also for some types of the Wegner models) it has been performed only for the density of states
(see \cite{Con:87}, \cite{DPS:02}), but not for the correlation function $R_k$. The
important step in studying of the second correlation was done in \cite{TSh:12}, where the behavior of the second
mixed moment of the characteristic polynomials was considered. It was proved that for 1D RBM
with $W^2\gg N$ this behavior (as $N\to\infty$) in the bulk of the spectrum coincides with
that for the GUE (this is closely related to (\ref{Un})). From the SUSY point of view
characteristic polynomials correspond to the so-called fermionic sector of the supersymmetric full model,
which describes the correlation functions $R_k$. In this paper we do the next step and present the
rigorous SUSY result about the second correlation function of block RBM (i.e. about the SUSY full model),
although with finite number of
blocks (which means that the width of the band is comparable with the matrix size).

\begin{theorem}\label{thm:1}
Let $\Lambda=[1,m]^d\cap \mathbb{Z}^d$ be a periodic
box, $H_N$ be the matrices (\ref{H}) -- (\ref{J}), $N=W|\Lambda|$, and let the number of sites $|\Lambda|$
be fixed. Then
\begin{equation*}
(N\rho(\lambda_0))^{-2}
R_2\left(\lambda_0+\displaystyle\frac{\xi_1}{\rho(\lambda_0)\,N},
\lambda_0+\displaystyle\frac{\xi_2}{\rho(\lambda_0)\,N}\right)\stackrel{w}{\longrightarrow}
1-\dfrac{\sin^2 (\pi(\xi_1-\xi_2))}
{\pi^2(\xi_1-\xi_2)^2},
\end{equation*}
as $W\to\infty$, for any $|\lambda_0|<\sqrt{2}$.
%\begin{equation}
%\pi^{-2}\lim\limits_{\varepsilon\to 0}\lim\limits_{W\to\infty}\left(\widetilde{G}_2(z,\hat{\xi})+
%\widetilde{G}_2(\bar{z},\hat{\xi})\right)=\dfrac{2}{\pi^2\rho^2(\lambda_0)}-
%\left(\dfrac{\sin(\pi(\xi_1-\xi_2))}{\pi(\xi_1-\xi_2)}\right)^2
%\end{equation}
\end{theorem}
\textit{Remark}

\begin{enumerate}
\item \textit{One can consider any finite regular graph instead of $\Lambda\subset \mathbb{Z}^d$.}

\item  \textit{The condition $|\Lambda|< \infty$ is not necessary here. For example, for $m$ which grows
like the small power of $W$ the proof can be repeated almost literally. Moreover, for $d=1$ the
same method is expected to work for $|\Lambda|\ll W^{1/2}$, i.e. $W\gg N^{2/3}$, but this requires
more delicate techniques.}

\item \textit{The condition $|\lambda_0|<\sqrt{2}$ is technical, the result should be the same
for any $\lambda_0\in (-2,2)$.}

\end{enumerate}

The paper is organized as follows. In Section $2$ we reformulate Theorem \ref{thm:1} in terms of the Green's
functions $G(z)$ and obtain a convenient integral
representation for $$\mathbf{E}\{G(z_1)\cdot G(z_2)\},\quad z_1,z_2\in\mathbb{C}$$ using the integration
over the Grassmann variables. Section $3$ deals with the preliminary results needed
for the proof. In Section $4$ we prove Theorem \ref{thm:1}, applying the steepest descent method
to the integral representation. Appendix is devoted to the introduction to the SUSY techniques.

%This model is called sometimes {\it the granular model}. From the physical point of view this model
%describes the situation, where each atom $j\in\Lambda$ is viewed as a grain (small metallic particle),
%having a large number $W$ of electron states, and the squared matrix elements for tunneling
%between grains ($J_{jk}$ for $j\ne k$) are small compared to the intra-grain matrix elements
%($J_{jj}$).

\subsection{Notation}
We denote by $C$, $C_1$, etc. various $W$-independent quantities below, which
can be different in different formulas. Integrals
without limits denote the integration (or the multiple integration) over the whole
real axis, or over the Grassmann variables.

Moreover,
\begin{itemize}

    \item[$\bullet$] $N=W|\Lambda|=W\cdot m^d;$

    \item[$\bullet$] $\mathbf{E}\{\ldots\}$ denotes the the expectation with respect to the
measure (\ref{pr_l});

    \item[$\bullet$] indices $i,j, j^\prime,k$ vary in $\Lambda$ and correspond to the number of the site
    (or the number of the block), index $l$ is always $1$ or $2$ (this is the field index),
    and Greek indices $\beta, \gamma$ vary from $1$ to $W$ and correspond to the position of
    the element in the block;

    \item[$\bullet$] big Latin letters (except $C$, which denotes different constants, and $W$) always denote
    $2\times 2$ matrices;

    \item[$\bullet$] variables $\phi$ and $\Phi$ with different indices are complex variables or vectors
    correspondingly;

    \item[$\bullet$] variables $\psi$ and $\Psi$ with different indices are Grassmann variables or vectors
    correspondingly;

    \item[$\bullet$] if $x_j$ means some variable which corresponds to the site $j\in \Lambda$, then $x$ means
    vector $\{x_j\}_{j\in\Lambda}$;

    \item[$\bullet$] $j\sim j^\prime$ means two adjacent points in $\Lambda$, the boundary conditions are periodic, i.e.
$$(p_1,\ldots,p_{i-1},1,p_{i+1},\ldots, p_m)\sim (p_1,\ldots,p_{i-1},m,p_{i+1},\ldots, p_m),$$
 and

    $$(\nabla x)^2=
    \sum\limits_{j\sim j^\prime} (x_j-x_{j^\prime})^2, \quad(\nabla X)^2=
    \sum\limits_{j\sim j^\prime} \Tr (X_j-X_{j^\prime})^2;$$

    \item[$\bullet$] $\Delta$ is a discrete Laplacian on $\Lambda$ with periodic boundary conditions, i.e.
    \[
    (\Delta x)_j=\sum\limits_{j^\prime: j\sim j^\prime} (x_j-x_{j^\prime});
    \]

  \item[$\bullet$] $J=\alpha\Delta/W+1/W,\quad \alpha<1/4d;$

  \end{itemize}

\begin{fleqn}[5pt]
    \begin{align}\label{a_pm}
    \bullet \,\,\, &L=\hbox{diag}\,\{1,-1\},\quad a_{\pm}=\dfrac{i\lambda_0\pm\sqrt{4-\lambda_0^2}}{2},\\
    \label{L_pm}
    &L_{\pm}=\hbox{diag}\,\{a_+,a_-\},\quad L_{\mp}=\hbox{diag}\,\{a_-,a_+\},\\ \notag
    &L_+=a_+I,\quad L_-=a_-I;
    \end{align}

    \begin{equation}\label{M_pm}
    \bullet \,\,\,M_\pm=\alpha a_{\pm}^2\Delta+(1+a_{\pm}^2)I;
    \end{equation}

    \end{fleqn}
\begin{itemize}
    \item[$\bullet$] $\mathring{U}(2)=U(2)/U(1)\times U(1)$, $\mathring{U}(1,1)=U(1,1)/U(1)\times U(1)$;

    %\item $\{U_j\}$, $\{P_j\}$, $\{V_j\}$, $\{\tilde{V}_j\}$, $\{\tilde{P}_j\}$ belong to $\mathring{U}$;
%
%    \item $\{T_j\}$, $\{S_j\}$ are matrices from the group $U(1,1)$, which is the group of $2\times 2$
%    matrices $T$ such that
%    $$T^*LT=L;$$

    %\item $\{B_j\}$ are matrices which can be obtained by multiplying
%    positive-definite $2\times 2$ matrices by $L$ (we denote this set by $\mathcal{H}_+L$);
%
%    \item $\rho_j$, $\tau_j$ are $2\times 2$ matrices whose entries are independent Grassmann
%    variables, $j\in\Lambda$;

    \item[$\bullet$] if $X$ means a matrix with eigenvalues $x_1$, $x_2$, then $\hat{X}=\hbox{diag}\,\{x_1,x_2\}$;

    \item[$\bullet$] $j=\overline{1}$ means $j=(1,\ldots,1)\in\Lambda$;

    \item[$\bullet$] $G(z)=(H_N-z)^{-1}$ is the Green's function of the matrix (\ref{H}) -- (\ref{J});

    \end{itemize}

    \begin{fleqn}[5pt]
 \begin{align}\label{z}
\bullet \,\,\, z_1&=\lambda_0+i\varepsilon/N+\xi_1/N\rho(\lambda_0),\quad z_2=
\lambda_0+i\varepsilon/N+\xi_2/N\rho(\lambda_0),\\ \notag
z_1^\prime&=\lambda_0+i\varepsilon/N+\xi_1^\prime/N\rho(\lambda_0),\quad z_2^\prime=
\lambda_0+i\varepsilon/N+\xi_2^\prime/N\rho(\lambda_0),
\end{align}
\quad \quad \,where $\lambda_0\in (-2,2)$, $\varepsilon>0$, and $\xi_1,\xi_2, \xi_1^\prime,
\xi_2^\prime\in [-C,C]\subset
\mathbb{R}$;

\begin{align}\label{G_2}
\bullet \,\,\, G_2^{+-}(z,\xi)&=\mathbf{E}\bigg\{\dfrac{\mdet(H_N-z_1^\prime)\mdet(H_N-\overline{z}_2)}
{\mdet(H_N-z_1)\mdet(H_N-\overline{z}_2^\prime)}\bigg\},\\ \notag
G_2^{++}(z,\xi)&=\mathbf{E}\bigg\{\dfrac{\mdet(H_N-z_1^\prime)\mdet(H_N-z_2)}
{\mdet(H_N-z_1)\mdet(H_N-z_2^\prime)}\bigg\}
\end{align}
\quad \quad \,for $z=(z_1,z_2)$ and $\xi=(\xi_1,\xi_2,\xi_1^\prime,\xi_2^\prime)$;

\begin{equation}\label{theta}
\bullet \,\,\,\theta_\varepsilon=-i\varepsilon+(\xi_2-\xi_1)/2\rho(\lambda_0), \quad c_0=
\sqrt{4-\lambda_0^2}=2\pi\rho(\lambda_0);
\end{equation}

\begin{equation}\label{delta}
\bullet \,\,\,\delta=\log W/\sqrt{W};
\end{equation}
\end{fleqn}
    \begin{itemize}

    \item[$\bullet$] $Z_s=\lambda_0I+i\varepsilon L/N+\hat{\xi}_s/N\rho(\lambda_0)$, $s=1,2$ or empty, where
    \begin{equation}\label{xi_hat}
\hat{\xi}=\left(\begin{array}{cc}
\xi_1&0\\
0& \xi_2
\end{array}\right),\quad \hat{\xi}_1=\left(\begin{array}{cc}
\xi_1^\prime&0\\
0& \xi_2
\end{array}\right)
, \quad\hat{\xi}_2=\left(\begin{array}{cc}
\xi_1&0\\
0& \xi_2^\prime
\end{array}\right).
\end{equation}

    \end{itemize}

\section{Integral representation}
According to the property of the Stieltjes transform, to prove Theorem \ref{thm:1}, it suffices to show that
\begin{multline}\label{main}
\dfrac{1}{(2\pi i N\rho(\lambda_0))^2}\lim\limits_{\varepsilon\to 0}\lim\limits_{W\to\infty}\mathbf{E}\Big\{\Tr
\Big(G(z_1)-G(\overline{z}_1)\Big)\cdot \Tr
\Big(G(z_2)-G(\overline{z}_2)\Big)\Big\}\\
=1-\dfrac{\sin^2 (\pi(\xi_1-\xi_2))}
{\pi^2(\xi_1-\xi_2)^2},
\end{multline}
where $G(z)$ is the resolvent of $H_N$.

Since
\begin{align}\notag
&(2\pi i N\rho(\lambda_0))^2 F_2(z_1,z_2):=\mathbf{E}\Big\{\Tr
\Big(G(z_1)-G(\overline{z}_1)\Big)\cdot \Tr
\Big(G(z_2)-G(\overline{z}_2)\Big)\Big\}\\ \label{F_2}
&=\mathbf{E}\Big\{\Tr
G(z_1)\cdot \Tr
G(z_2)\Big\}+\mathbf{E}\Big\{\overline{\Tr
G(z_1)\cdot \Tr
G(z_2)}\Big\}\\ \notag
&-\mathbf{E}\Big\{\Tr
G(z_1)\cdot \Tr
G(\overline{z}_2)\Big\}-\mathbf{E}\Big\{\overline{\Tr
G(z_1)\cdot \Tr
G(\overline{z}_2)}\Big\},
\end{align}
we get
\begin{multline}\label{F_expr}
F_2(z_1,z_2)=(2\pi)^{-2}\dfrac{\partial^2}{\partial \xi_1^\prime \partial \xi_2^\prime}\Big(
G_2^{++}(z,\xi)+\overline{G_2^{++}}(z,\xi)\\-G_2^{+-}(z,\xi)-\overline{G_2^{+-}}(z,\xi)\Big)
\Big|_{\xi^\prime=\xi},
\end{multline}
where $\xi^\prime=\xi$ means $\xi_1^\prime=\xi_1$, $\xi_2^\prime=\xi_2$, and $G_2^{++}$, $G_2^{+-}$ are
defined in (\ref{G_2}).

Thus, we have to find an integral representation for $$\mathbf{E}\Big\{\Tr
G(z_1)\cdot \Tr
G(z_2)\Big\}\,\,\,\hbox{and}\,\,\,\mathbf{E}\Big\{\Tr
G(z_1)\cdot \Tr
G(\overline{z}_2)\Big\}.$$

Note that since the density of states for (\ref{H}) is (\ref{rho}), we have
\begin{multline}\label{pr_St}
\lim\limits_{\varepsilon\to 0}\lim\limits_{W\to \infty}\dfrac{1}{N\rho(\lambda_0)}
\mathbf{E}\Big\{\Tr G(\lambda_0+i\varepsilon)\Big\}\\=\dfrac{1}{\rho(\lambda_0)}\lim\limits_{\varepsilon\to 0}
\int\dfrac{\rho(\lambda)d\lambda}{\lambda-\lambda_0-i\varepsilon}=\dfrac{-\lambda_0+i\sqrt{4-\lambda_0^2}}
{2\rho(\lambda_0)}.
\end{multline}
Hence, we obtain by the construction (see (\ref{G_2}), (\ref{pr_St}), and (\ref{a_pm}))
\begin{align}\label{Ward_id}
&G_2^{+-}(z,\xi)\Big|_{\xi^\prime=\xi}=
G_2^{+-}(z,\xi)\Big|_{\xi^\prime=\xi}=1,\\ \notag
&\dfrac{\partial}{\partial \xi_1^\prime}G_2^{+-}(z,\xi)\Big|_{\xi^\prime=\xi}=
\dfrac{\partial}{\partial \xi_1^\prime}G_2^{++}(z,\xi)\Big|_{\xi^\prime=\xi}\\ \notag
&=-\dfrac{1}{N\rho(\lambda_0)}\mathbf{E}\Big\{\Tr G(z_1)\Big\}=-ia_+/\rho(\lambda_0)+o(1),\\
\notag
&\dfrac{\partial}{\partial \xi_2^\prime}G_2^{++}(z,\xi)\Big|_{\xi^\prime=\xi}=
\dfrac{1}{N\rho(\lambda_0)}\mathbf{E}\Big\{\Tr G(z_2)\Big\}=ia_+/\rho(\lambda_0)+o(1),\\ \notag
&\dfrac{\partial}{\partial \xi_2^\prime}G_2^{+-}(z,\xi)\Big|_{\xi^\prime=\xi}=\dfrac{1}{N\rho(\lambda_0)}
\mathbf{E}\Big\{\Tr \overline{G(z_2)}\Big\}=ia_-/\rho(\lambda_0)+o(1).
\end{align}

We are going to obtain the integral representations for $G_2^{++}(z,\xi)$ and $G_2^{+-}(z,\xi)$
by using rather standard SUSY techniques, i.e. integrals
over the Grassmann variables. Integration over the Grassmann variables was introduced by
Berezin (see \cite{Ber}) and is widely used in the physics literature (see e.g. \cite{Ef} and \cite{M:00}).
Here we use the modification of the method which uses the superbozonization formula (see \cite{SupB:08}).
For the reader convenience we give
a brief outline of the techniques in Appendix.

This method allows us to obtain
the formula for products and ratios of the characteristic
polynomials which is very useful for the
averaging because it is a Gaussian-type integral (see formulas (\ref{G_C}) -- (\ref{G_Gr}) below).
After averaging over the probability measure we can integrate over the Grassmann variables to obtain
an integral representation (in complex variables) which can be studied by the steepest descent method.

Set
\begin{align}\label{P_n}
\mathcal{P}_N(U,B)=\intd&\exp\Big\{\alpha\sum\limits_{j\sim j^\prime}
\Tr (\rho_j-\rho_{j^\prime})
(\tau_j-\tau_{j^\prime})-\sum\limits_{j\in\Lambda} \Tr \rho_j\tau_j\Big\}\\ \notag
&\times\prod\limits_{j\in\Lambda}\mdet^{-W} (1+W^{-1}U_j^{-1}\rho_jB_j^{-1}\tau_j)
\prod\limits_{j\in\Lambda}
d\rho_jd\tau_j,
\end{align}
where $\rho_j$, $\tau_j$ are $2\times 2$ matrices whose entries are independent Grassmann variables,
$U_j\in U(2)$, $B_j\in \mathcal{H}_+L$,
and let
\begin{align}\label{L_cal}
\mathcal{L}_\pm(\lambda_0)&=\Big\{r\Big(\pm i\lambda_0/2+\sqrt{4-\lambda_0^2}/2\Big)|r\in [0,+\infty)\Big\}.
\end{align}
Introduce
\begin{align}\label{K}
K_{m}(V,\hat{U})=&\alpha (\nabla \hat{U})^2/2+\alpha\sum\limits_{j\sim j^\prime} |(V_{j^\prime}V_j^*)_{12}|^2
(u_{j^\prime,1}-u_{j^\prime,2})(u_{j,1}-u_{j,2})
\\ \notag
&-
\sum\limits_{j\in\Lambda}\big(\Tr \hat{U}_j^2/2-i\lambda_0 \Tr \hat{U}_j-
\log \mdet \hat{U}_j\big)-U_*,
\end{align}
\begin{align}\notag
L_{m}(T,\hat{B})=&-\alpha (\nabla \hat{B})^2/2+\alpha\sum\limits_{j\sim j^\prime} |(T_{j^\prime}T_j^{-1})_{12}|^2
(b_{j^\prime,1}+b_{j^\prime,2})(b_{j,1}+b_{j,2})\\ \label{L}
&+
\sum\limits_{j\in\Lambda}\big(\Tr \hat{B}_j^2/2-i\lambda_0 \Tr \hat{B}_j-
\log\mdet\hat{B}_j\big)+U_*
\end{align}
with some constant $U_*$ that will be chosen later (see (\ref{U*})). Here $V_j\in \mathring{U}(2)$, $V_1=I$,
$T_j\in \mathring{U}(1,1)$, $T_1=I$, $b_{j,1}, b_{j,2}\in \mathbb{R}_+$, $u_{j,1}, u_{j,2}\in \mathbb{T}$, and
\[
\hat{U}_j=\left(\begin{array}{cc}
u_{j,1}&0\\
0&u_{j,2}
\end{array}\right),\quad \hat{B}_j=\left(\begin{array}{cc}
b_{j,1}&0\\
0&-b_{j,2}
\end{array}\right).
\]
Define also
\begin{align}\notag
\mathcal{F}_m(P_1, S_1, V, T,\hat{U},\hat{B})&=\dfrac{\varepsilon}{|\Lambda|}
\sum\limits_{j\in\Lambda} \Tr (V_jP_1^*)^*\hat{U}_j(V_jP_1^*)L
\\ \label{F_cal}&-\dfrac{\varepsilon}{|\Lambda|}
\sum\limits_{j\in\Lambda} \Tr (T_jS_1^{-1})^{-1}\hat{B}_j(T_jS_1^{-1})L
\\ \notag
&-\dfrac{i}{|\Lambda|\rho(\lambda_0)}\sum\limits_{j\in\Lambda}\Tr (V_jP_1^*)^*\hat{U}_j(V_jP_1^*)
\hat{\xi}_1\\ \notag &+\dfrac{i}{|\Lambda|\rho(\lambda_0)}\sum\limits_{j\in\Lambda}
\Tr (T_jS_1^{-1})^{-1}\hat{B}_j(T_jS_1^{-1})\hat{\xi}_2,
\end{align}
where $|\Lambda|=m^d$, $\rho(\lambda_0)$ is defined in (\ref{rho}), and
$\hat{\xi}_{1,2}$ is defined in (\ref{xi_hat}).
The purpose of this section is to prove
\begin{proposition}\label{p:int_rep_1}
The function $G_2^{+-}(z,\xi)$ of (\ref{G_2}) can be represented as follows:
\begin{align}\notag
&G_2^{+-}(z,\xi)=\dfrac{W^{4|\Lambda|}}{(8\pi^2)^{|\Lambda|}}\displaystyle\int \prod\limits_{j\in\Lambda\setminus\{\overline{1}\}} d\nu(T_j)d\mu(V_j)
\oint_{\mathbb{T}^{2|\Lambda|}} \prod\limits_{j\in\Lambda}
d u_{j,1}du_{j,2}\\ \notag
&\int_{\mathcal{L}_+(\lambda_0)^{|\Lambda|}} \prod\limits_{j\in\Lambda}d b_{j,1}
\int_{\mathcal{L}_-(\lambda_0)^{|\Lambda|}
}\prod\limits_{j\in\Lambda}db_{j,2}
\cdot \exp\Big\{-W\left(K_{m}(V,\hat{U})+L_{m}(T,\hat{B})\right)\Big\}\\ \label{G_int}
&\times\mathcal{P}_N
\left(V^*\hat{U}V,T^{-1}\hat{B}T\right) \prod\limits_{j\in\Lambda}
(u_{j,1}-u_{j,2})^2\,\prod\limits_{j\in\Lambda}
(b_{j,1}+b_{j,2})^2\\ \notag
&\times \int d\nu(S_1)d\mu(P_1)  \exp\Big\{\mathcal{F}_m(P_1, S_1, V,T,\hat{U},\hat{B})\Big\},
\end{align}
where $d\mu$, $d\nu$ are the Haar measures over $\mathring{U}(2)$ and $\mathring{U}(1,1)$ correspondingly, which
can be parameterized as follows (see (\ref{mu}) -- (\ref{nu}))
\begin{align}\label{mu_j}
&V_j=\left(\begin{array}{ll}
w_j& v_j\,e^{i\theta_{j}}\\
-v_j\,e^{-i\theta_{j}}& w_j\\
\end{array}\right),\quad
w_j=(1-v_j^2)^{1/2}\\ \notag
& d\mu(V_j)=\dfrac{d\theta_{j}}{2\pi}\cdot
(2v_j dv_j),\quad v_j\in [0,1],\quad \theta_{j}\in [0,2\pi],
\end{align}
\begin{align}\label{nu_j}
&T_j=\left(\begin{array}{ll}
s_j& t_j\,e^{i\sigma_j}\\
t_j\,e^{-i\sigma_{j}}& s_j\\
\end{array}\right),\quad
s_j=(1+t_j^2)^{1/2}\\ \notag
& d\mu(T_j)=\dfrac{d\sigma_{j}}{2\pi}\cdot
(2t_j dt_j),\quad t_j\in [0,\infty),\quad \sigma_{j}\in [0,2\pi].
\end{align}

\end{proposition}
\begin{proof}
Introduce complex and Grassmann fields:
\begin{align*}
\Phi_l=\{\phi_{jl}\}^t_{j\in\Lambda},&\quad \phi_{j l}=(\phi_{j l 1}, \phi_{j l 2},\ldots,
\phi_{j l W}),\quad l=1,2,
\quad -\quad \hbox{complex},\\
\Psi_l=\{\psi_{jl}\}^t_{j\in\Lambda}, &\quad \psi_{j l}=(\psi_{j l 1}, \psi_{j l 2},\ldots,
\psi_{j l W}),\quad l=1,2, \quad -\quad \hbox{Grassmann}.
\end{align*}
Using (\ref{G_C}) -- (\ref{G_Gr}) (see Appendix) we obtain
\begin{equation*}
\begin{array}{c}
G_2^{+-}(z,\xi)
=\mathbf{E}\Big\{\displaystyle\int \exp\{i\Psi_1^+(z_1^\prime-H_N)\Psi_1
-i\Psi_2^+(\overline{z}_2-H_N)\Psi_2\}\\
\times\exp\{i\Phi_1^+(z_1-H_N)\Phi_1-i\Phi_2^+
(\overline{z}_2^\prime-H_N)\Phi_2\}d\Phi d\Psi\Big\}\\
=\displaystyle\int d\Phi d\Psi\,\, \exp\Big\{i(z_1^\prime\Psi_1^+\Psi_1
+z_1\Phi_1^+\Phi_1)-i(\overline{z}_2\Psi_2^+\Psi_2
+\overline{z}_2^\prime\Phi_2^+\Phi_2)\Big\}\\
\times\mathbf{E}\Big\{\exp\Big\{-\sum\limits_{j\le k}\sum\limits_{\alpha, \beta}
\Big(i\Re H_{jk,\alpha\beta}\chi^+_{jk,\alpha\beta}
-\Im H_{jk,\alpha\beta}\chi^-_{jk,\alpha\beta}\Big)\Big\}\Big\},
\end{array}
\end{equation*}
where $z_l, z_l^\prime$ are defined in (\ref{z}),
\begin{align*}
&d\Phi=\prod\limits_{j\in\Lambda}\prod\limits_{\alpha=1}^W\prod\limits_{l=1}^2\dfrac{d\Re\phi_{jl\alpha}
d\Im\phi_{j l\alpha}}{\pi},\quad
d\Psi=\prod\limits_{j\in\Lambda}\prod\limits_{\alpha=1}^W\prod\limits_{l=1}^2d\overline{\psi}_{jl\alpha}
d\psi_{j l\alpha},\\
&\chi^{\pm}_{jk,\alpha\beta}=\eta_{jk,\alpha\beta}\pm \eta_{kj,\beta\alpha},\\
&\eta_{jk,\alpha\beta}=\overline{\psi}_{j 1\alpha}\psi_{k 1\beta}-
\overline{\psi}_{j 2\alpha}\psi_{k 2\beta}+\overline{\phi}_{j 1\alpha}\phi_{k 1\beta}-
\overline{\phi}_{j 2\alpha}\phi_{k 2\beta},\\
&\eta_{jj,\alpha\alpha}=(\overline{\psi}_{j 1\alpha}\psi_{j 1\alpha}-
\overline{\psi}_{j 2\alpha}\psi_{j 2\alpha}+\overline{\phi}_{j 1\alpha}\phi_{j 1\alpha}-
\overline{\phi}_{j 2\alpha}\phi_{j 2\alpha})/2.
\end{align*}
Averaging over (\ref{pr_l}), we get
\begin{equation*}
\begin{array}{c}
G_2^{+-}(z,\xi)=\displaystyle\int d\Phi d\Psi\,\, \exp\Big\{i(z_1^\prime\Psi_1^+\Psi_1
+z_1\Phi_1^+\Phi_1)-i(\overline{z}_2\Psi_2^+\Psi_2
+\overline{z}_2^\prime\Phi_2^+\Phi_2)\Big\}\\
\times\exp\Big\{-\sum\limits_{j<k, \alpha,\beta} J_{jk}\,\,
\eta_{jk,\alpha\beta}\eta_{kj,\beta\alpha}-\frac{1}{2}\sum\limits_{j, \alpha} J_{jj}\,\,
\eta_{jj,\alpha\alpha}^2\Big\}.
\end{array}
\end{equation*}
Thus, we have
\begin{equation}\label{G_av}
\begin{array}{c}
G_2^{+-}(z,\xi)=\displaystyle\int d\Phi d\Psi\,\, \exp\Big\{i\sum\limits_{j\in \Lambda} \Tr \tilde{X}_jLZ_1+
i\sum\limits_{j\in \Lambda} \Tr \tilde{Y}_jLZ_2\Big\}\\
\times \exp\Big\{\dfrac{1}{2}\sum\limits_{j,k\in\Lambda}J_{jk}\Tr (\tilde{X}_jL)(\tilde{X}_kL)-
\dfrac{1}{2}\sum\limits_{j,k\in\Lambda}J_{jk}\Tr (\tilde{Y}_jL)(\tilde{Y}_kL)\Big\}\\
\times \exp\Big\{-\sum\limits_{j,k\in\Lambda}J_{jk}\Tr (\tilde{\rho}_jL)(\tilde{\tau}_kL)\Big\},
\end{array}
\end{equation}
where $L$ is defined in (\ref{L_pm}),
$$Z_{1,2}=\lambda_0I+i\varepsilon L/N+\hat{\xi}_{1,2}/N\rho(\lambda_0),$$
\begin{align*}
\tilde{X}_j=\left(
\begin{array}{ll}
\psi_{j1}^+\psi_{j1}& \psi_{j1}^+\psi_{j2}\\
\psi_{j2}^+\psi_{j1}& \psi_{j2}^+\psi_{j2}
\end{array}
\right),& \quad \tilde{Y}_j=\left(
\begin{array}{ll}
\phi_{j1}^+\phi_{j1}& \phi_{j1}^+\phi_{j2}\\
\phi_{j2}^+\phi_{j1}& \phi_{j2}^+\phi_{j2}
\end{array}
\right),\\
\tilde{\rho}_j=\left(
\begin{array}{ll}
\psi_{j1}^+\phi_{j1}& \psi_{j1}^+\phi_{j2}\\
\psi_{j2}^+\phi_{j1}& \psi_{j2}^+\phi_{j2}
\end{array}
\right), &\quad \tilde{\tau}_j=\left(
\begin{array}{ll}
\phi_{j1}^+\psi_{j1}& \phi_{j1}^+\psi_{j2}\\
\phi_{j2}^+\psi_{j1}& \phi_{j2}^+\psi_{j2}
\end{array}
\right).
\end{align*}
Applying the superbosonization formula (see Proposition \ref{p:supboz}), we obtain
\begin{align}\notag
G_2^{+-}(z,\xi)&=(-\pi^2)^{-|\Lambda|}\displaystyle\int \prod\limits_{j\in \Lambda} dX_jdY_j
\prod\limits_{j\in \Lambda} d\rho_jd\tau_j  \\
&\times \exp\Big\{i\sum\limits_{j\in \Lambda}\Tr X_jLZ_1+
i\sum\limits_{j\in \Lambda} \Tr Y_jLZ_2\Big\}\label{sup}\\ \notag
&\times \exp\Big\{\dfrac{1}{2}\sum\limits_{j,k\in\Lambda}J_{jk}\Tr (X_jL)(X_kL)-
\dfrac{1}{2}\sum\limits_{j,k}J_{jk}\Tr (Y_jL)(Y_kL)\Big\}\\ \notag
&\times \exp\Big\{-\sum\limits_{j,k\in\Lambda}J_{jk}\Tr (\rho_jL)(\tau_kL)\Big\}
\prod\limits_{j\in\Lambda}\dfrac{\mdet^W Y_j}{\mdet^W (X_j-\rho_jY_j^{-1}\tau_j)},
\end{align}
where $\{X_j\}_{j\in\Lambda}$ are unitary $2\times 2$ matrices, $\{Y_j\}_{j\in\Lambda}$ are the
positive Hermitian matrices,
$\{\rho_j\}_{j\in\Lambda}$, $\{\tau_j\}_{j\in\Lambda}$ are $2\times 2$ matrices with
independent Grassmann variables, and
%\begin{align*}
%%dX_j&=(u_{j,1}-u_{j,2})^2d\,\mu(V_j)\prod\limits_{l=1}^2\dfrac{du_{l,j}}{2\pi i},\\
%dY_j&=\mathbf{1}_{Y_j>0}\cdot d\Re Y_{12,j}\,d\Im Y_{12,j}\,d Y_{11,j}\,d Y_{22,j},\\
%d \sigma_jd\tau_j&=\prod\limits_{l=1}^2\prod\limits_{s=1}^2d \sigma_{ls,j}d\tau_{ls,j},
%\end{align*}
$dX_j$, $dY_j$, $d\rho_j\,d\tau_j$ is defined in Proposition \ref{p:supboz}.
%where $u_{j,1}, u_{j,2}$ are the eigenvalues of $U_j$, $V_j$ is a matrix which diagonalizes $U_j$,
%$d\,u_j$ means the integration over the unit
%circle $\mathbb{T}=\{z:|z|=1\}$ and $d\,\mu(V_j)$ is the Haar measure
%over the unitary group $U(2)$.

  Shifting $\rho_jL\to \sqrt{W} \rho_j$, $\tau_jL\to \sqrt{W}\tau_j$ and defining
  $B_j=W^{-1} Y_jL$,\\ $U_j=-W^{-1}X_jL$ as new
variables, we can rewrite (\ref{sup}) as
\begin{align}\notag
G_2^{+-}(z,\xi)&=\dfrac{W^{4|\Lambda|}}{(-\pi^2)^{|\Lambda|}}\displaystyle\int \prod\limits_{j\in\Lambda} dU_jdB_j
\prod\limits_{j\in\Lambda} d\rho_jd\tau_j\\ \notag
&\times  \exp\Big\{-iW\sum\limits_{j\in\Lambda} \Tr U_jZ_1+
iW\sum\limits_{j\in\Lambda} \Tr B_jZ_2\Big\}\\ \label{sup1}
&\times \exp\Big\{\dfrac{W^2}{2}\sum\limits_{j,k}J_{jk}\Tr U_jU_k-
\dfrac{W^2}{2}\sum\limits_{j,k}J_{jk}\Tr B_jB_k\Big\}\\ \notag
&\times \exp\Big\{-W\sum\limits_{j,k}J_{jk}\Tr \rho_j\tau_k\Big\}
\prod\limits_{j\in\Lambda}\dfrac{\mdet^W B_j}{\mdet^W (U_j+W^{-1}\rho_jB_j^{-1}\tau_j)},
\end{align}
where $dU_j$ is defined in Proposition \ref{p:supboz}, and
$$dB_j=\mathbf{1}_{\widetilde{B}_jL>0}\cdot d\Re B_{12,j}\,d\Im B_{12,j}\,d B_{11,j}\,
d B_{22,j}.$$
Change the variables to
\begin{align}\label{diag}
U_j&=P_j^*\hat{U}_jP_j,\,\,\,\,\,\, \hat{U}_j=\hbox{diag}\,\{u_{j,1},u_{j,2}\},\,\,\,\,  P_j\in \mathring{U}(2),
\quad\,\,\, u_{j,1},u_{j,2}\in \mathbb{T},\\ \notag
B_j&=S_j^{-1}\hat{B}_jS_j,\,\, \hat{B}_j=\hbox{diag}\,\{b_{j,1},-b_{j,2}\},\,\,
S_j\in \mathring{U}(1,1),
\,\,\, b_{j,1},b_{j,2}\in \mathbb{R}^+,
\end{align}
and then
\begin{align*}
V_j&=P_jP_1^*,\quad\,\,\, j\in\Lambda\setminus\{\overline{1}\},\quad V_j\in \mathring{U}(2),\,\,\,\,\,\quad V_1=I,\\ \notag
T_j&=S_jS_1^{-1},\quad j\in\Lambda\setminus\{\overline{1}\},\quad T_j\in \mathring{U}(1,1),\quad T_1=I.
\end{align*}
The Jacobian of such a change is
\[
2^{|\Lambda|}(\pi/2)^{2|\Lambda|}\prod\limits_{j\in\Lambda}(u_{j,1}-u_{j,2})^2
\prod\limits_{j\in\Lambda}(b_{j,1}+b_{j,2})^2
\]
Substituting this, the expressions for $d U_j$, $d B_j$, and (\ref{J}), we obtain from~(\ref{sup1})
\begin{align}\label{G_last}
G_2^{+-}(z,\xi)&=\dfrac{W^{4|\Lambda|}}{(8\pi^2)^{|\Lambda|}}\displaystyle
\int \prod\limits_{j\in\Lambda\setminus\{\overline{1}\}} d\nu(T_j)d\mu(V_j)
\oint_{\mathbb{T}^{2|\Lambda|}} \prod\limits_{j\in\Lambda} d u_{j,1}du_{j,2}\\ \notag
&\times \int_{\mathbb{R}_+^{2|\Lambda|}} \prod\limits_{j\in\Lambda} d b_{j,1}db_{j,2}
\exp\Big\{-W\left(K_{m}(V,\hat{U})+L_{m}(T,\hat{B})\right)\Big\}\\ \notag
&\times\mathcal{P}_N
\left(V^*\hat{U}V, T^{-1}\hat{B}T\right) \prod\limits_{j\in\Lambda}
(u_{j,1}-u_{j,2})^2\,\prod\limits_{j\in\Lambda}
(b_{j,1}+b_{j,2})^2\\ \notag
&\times\int d\nu(S_1)d\mu(P_1) \exp\Big\{\mathcal{F}_m(P_1,S_1,V,T,\hat{U},\hat{B})\Big\},
\end{align}
where $d\mu$, $d\nu$ are the Haar measures of (\ref{mu}) and (\ref{nu}) correspondingly.

Here $\mathcal{P}_N$ and $\mathcal{F}_m$ are defined in (\ref{P_n}) and (\ref{F_cal}), and
\begin{multline*}
K_{m}(V,\hat{U})=\frac{\alpha}{2}\left(\nabla (V^*\hat{U}V)\right)^2\\
-
\sum\limits_{j\in\Lambda}\big(\Tr \hat{U}_j^2/2-i\lambda_0 \Tr \hat{U}_j-
\log \mdet \hat{U}_j\big)-U_*,
\end{multline*}
\begin{multline*}
L_{m}(S,\hat{B})=-\frac{\alpha}{2}\left(\nabla (T^{-1}\hat{B}T)\right)^2\\
+
\sum\limits_{j\in\Lambda}\big(\Tr \hat{B}_j^2/2-i\lambda_0 \Tr \hat{B}_j-
\log\mdet\hat{B}_j\big)+U_*,
\end{multline*}
where $U_*$ is a constant that will be chosen below (see (\ref{U*})).
Rewriting
\begin{align*}
\Tr (V_j^*\hat{U}_jV_j-V_{j^\prime}^*
\hat{U}_{j^\prime}V_{j^\prime})^2
=&\Tr (\hat{U}_j-\hat{U}_{j^\prime})^2\\ \notag&+
|(V_{j^{\prime}}V_j^*)_{12}|^2(u_{j,1}-u_{j,2})(u_{j^\prime,1}-u_{j^\prime,2}),\\ \notag
\Tr (T_j^{-1}\hat{B}_jT_j
-T_{j^\prime}^{-1}\hat{B}_{j^\prime}T_{j^\prime})^2=&\Tr (\hat{B}_j-\hat{B}_{j^\prime})^2\\ \notag
&-|(T_{j^\prime} T_j^{-1})_{12}|^2(b_{j,1}+b_{j,2})(b_{j^\prime,1}+b_{j^\prime,2}),
\end{align*}
we get (\ref{K}) and (\ref{L}).
%\begin{align*}\notag
%&G_2^{+-}(z,\xi)=\dfrac{W^{4|\Lambda|}}{(8\pi^2)^{|\Lambda|}}\displaystyle\int \prod\limits_{j\in\Lambda\setminus\{\overline{1}\}} d\nu(T_j)d\mu(V_j)
%\oint_{\mathbb{T}^{2|\Lambda|}} \prod\limits_{j\in\Lambda}
%d u_{j,1}du_{j,2}\int_{\mathbb{R}_+^{2|\Lambda|}} d b_{j,1}db_{j,2}\\
%&\times \exp\Big\{-W\left(K_{m}(V,\hat{U})+L_{m}(T,\hat{B})\right)\Big\}\cdot\mathcal{P}_N
%\left(V^*\hat{U}V,T^{-1}\hat{B}T\right)\\ \notag
%&\times \prod\limits_{j\in\Lambda}
%(u_{j,1}-u_{j,2})^2\,\prod\limits_{j\in\Lambda}
%(b_{j,1}+b_{j,2})^2 \int d\nu(S_1)d\mu(P_1)  \exp\Big\{\mathcal{F}_m(P_1, S_1, V,T,\hat{U},\hat{B})\Big\},
%\end{align*}

Finally, note that for $|\lambda_0|<\sqrt{2}$ we can move the contour of integration over $\{b_{j,1}\}$,
$\{b_{j,2}\}$ in (\ref{G_last})
from $\mathbb{R}_+$ to $\mathcal{L}_+(\lambda_0)$ and $\mathcal{L}_-(\lambda_0)$ respectively.

Indeed, set $I_R=\{z\in\mathbb{C}: z=R+ix,\,\,x\in \mathbb{R}\}$ and consider the contours
\begin{align*}
\mathcal{C}_R^{+}&=(0,R)\cup\{z\in\mathcal{L}_+(\lambda_0): 0\le \Re z\le R\}\\
&\cup
\{z\in\mathbb{C}: z=R+ix,\,\,x\in \mathbb{R},0\le \arg z\le \arg a_+\},\\
\mathcal{C}_R^{-}&=\{\overline{z}: z\in\mathcal{C}_R^{+}\}.
\end{align*}
Take $b_{j_0,1}$ and fix all other $b_{j,1}\in \mathbb{R}_+\cup\mathcal{L}_+(\lambda_0)$, and
$b_{j,2}\in \mathbb{R}_+\cup\mathcal{L}_-(\lambda_0)$. Since the integrand is analytic,
integrating with respect to $b_{j_0,1}$ over $\mathcal{C}_R^{+}$ we get $0$. Note also
that
\[
\Re (b_{j_0,1}+b_{j_0,2})(b_{j,1}+b_{j,2})\ge 0
\]
for any $b_{j_0,1}\in \mathcal{C}_R^{+}\cap I_R$ and for any $j\in\Lambda$, since
both brackets have arguments from $-\pi/4$ to $\pi/4$ (because $|\lambda_0|< \sqrt{2}$ and thus
$|\arg a_\pm|< \pi/4$ (see (\ref{L_pm}))). Thus, (\ref{L_rew}) -- (\ref{B_pm}) yield
\begin{align*}
\Re L_m(T,\hat{B})\ge CR^2
\end{align*}
for sufficiently big $R$, i.e. the integral with respect to $b_{j_0,1}$ over
$\mathcal{C}_R^{+}\cap I_R$ tends to $0$, as $R\to\infty$.
Hence, we can change the contour of integration over $b_{j_0,1}$ from  $\mathbb{R}_+$ to
$\mathcal{L}_+(\lambda_0)$. Repeating the procedure for all $j_0\in\Lambda$, we get (\ref{G_int}).~$\quad\Box$
\end{proof}

%and change $s_j$ of (\ref{nu}) to $s_j(b_{j,1}+b_{j,2})^{-1/2}(b_{j^\prime,1}+b_{j^\prime,2})^{-1/2}$.
%We get then
%\begin{align}\notag
%G_2(z,\xi)&=\dfrac{W^{4|\Lambda|}}{(2\pi)^{4|\Lambda|}}\displaystyle\int \prod\limits_{j\in\Lambda} d\mu(V_j)d \nu(S_j)
%\oint \prod\limits_{j\in\Lambda} \dfrac{d u_{j,1}du_{j,2}}{(2\pi i)^2}\int_{0}^\infty d b_{j,1}db_{j,2}\\ \notag
%&\times \exp\Big\{-W\Big(K_{m}(V,\hat{U})+L_{m}(\hat{B})+\alpha\sum_{j=2}^m s_j^2\Big)+\mathcal{F}_m(V,S,\hat{U},\hat{B})
%\Big\}\\ \label{G_2_2}
%&\times F_m\Big(\{(P_jV_1)^*\hat{U}_j(P_jV_1)\},\{(Q_jS_1)^{-1}\hat{B}_j(Q_jS_1)\}\Big)
%\\ \notag &\times
%\prod\limits_{j\in\Lambda}
%(u_{j,1}-u_{j,2})^2\,(b_{1,1}+b_{2,1})(b_{1,m}+b_{2,m}),
%\end{align}
%where
%\begin{align}\label{L}
%&L_{m}(S)=-\dfrac{\alpha}{2}\sum\limits_{j=2}^m \Tr (\hat{B}_j-
%\hat{B}_{j^\prime})^2+
%\sum\limits_{j\in\Lambda}\big(\dfrac{1}{2}\Tr \hat{B}_j^2-i\lambda_0 \Tr \hat{B}_j-
%\log\mdet\hat{B}_j\big)+U_*,
%\end{align}
%$S_j$ is such as in (\ref{nu}) but with $s_j\to s_j(b_{j,1}+b_{j,2})^{1/2}(b_{j^\prime,1}+b_{j^\prime,2})^{1/2}$,
%and $V_j$, $d \mu(V_j)$, $d\nu (S_j)$ are defined in (\ref{mu}) -- (\ref{nu}).
\begin{proposition}\label{p:int_rep_2}
The function $G_2^{++}(z,\xi)$ of (\ref{G_2}) can be represented as follows:
\begin{align*}\notag
G_2^{++}(z,\xi)&=\dfrac{W^{4|\Lambda|}}{(4\pi)^{2|\Lambda|}}\displaystyle\int
\prod\limits_{j\in\Lambda\setminus\{\overline{1}\}} d\mu(\tilde{V}_j)d\mu(V_j)
\oint_{\mathbb{T}^{2|\Lambda|}} \prod\limits_{j\in\Lambda}
d u_{j,1}du_{j,2}\\
&\int_{\mathcal{L}_+(\lambda_0)^{2|\Lambda|}} \prod\limits_{j\in\Lambda}d a_{j,1}da_{j,2}\cdot
\exp\Big\{-W\left(K_{m}(V,\hat{U})+\tilde{L}_{m}(\tilde{V},\hat{A})\right)\Big\}\\ \notag
&\times \mathcal{P}_N
\left(P_1,\tilde{P}_1,V^*\hat{U}V,\tilde{V}^*\hat{A}\tilde{V}\right)
\prod\limits_{j\in\Lambda}
(u_{j,1}-u_{j,2})^2\,\prod\limits_{j\in\Lambda}
(a_{j,1}-a_{j,2})^2 \\ \notag
&\times\int d\mu(\tilde{P}_1)d\mu(P_1)
  \exp\Big\{\mathcal{F}_m(P_1,
\tilde{P}_1, V,\tilde{V},\hat{U},\hat{A})\Big\},
\end{align*}
where $d\mu$ is defined in (\ref{mu_j}), $P_1,\tilde{P}_1\in \mathring{U}(2)$, and
\begin{align}\label{L_til}
\tilde{L}_{m}(\tilde{V},\hat{A})=&-\alpha (\nabla \hat{A})^2/2\\ \notag
&-\alpha\sum\limits_{j\sim j^\prime}
|(\tilde{V}_{j^\prime}\tilde{V}_j^*)_{12}|^2
(a_{j^\prime,1}-a_{j^\prime,2})(a_{j,1}-a_{j,2})\\ \notag
&+
\sum\limits_{j\in\Lambda}\big(\Tr \hat{A}_j^2/2-i\lambda_0 \Tr \hat{A}_j-
\log\mdet\hat{A}_j\big)+U_*
\end{align}
with some constant $U_*$ that will be chosen later (see (\ref{U*})). Here
$V_j,\tilde{V}_j\in \mathring{U}(2)$, $V_1=\tilde{V}_1=I$,
$a_{j,1}, a_{j,2}\in \mathcal{L}_+(\lambda_0)$, $u_{j,1}, u_{j,2}\in \mathbb{T}$, and
\[
\hat{U}_j=\left(\begin{array}{cc}
u_{j,1}&0\\
0&u_{j,2}
\end{array}\right),\quad \hat{A}_j=\left(\begin{array}{cc}
a_{j,1}&0\\
0&a_{j,2}
\end{array}\right).
\]
\end{proposition}
\begin{proof}
Again using (\ref{G_C}) -- (\ref{G_Gr}) we get
\begin{equation*}
\begin{array}{c}
G_2^{++}(z,\xi)
=\mathbf{E}\Big\{\displaystyle\int \exp\{i\Psi_1^+(z_1^\prime-H_N)\Psi_1
+i\Psi_2^+(z_2-H_N)\Psi_2\}\\
\times\exp\{i\Phi_1^+(z_1-H_N)\Phi_1+i\Phi_2^+
(z_2^\prime-H_N)\Phi_2\}d\Phi d\Psi\Big\}.
\end{array}
\end{equation*}
After averaging this gives the r.h.s. of (\ref{G_av}), but without $L$.
Further calculations repeat almost literally the proof of Proposition 1.
The only difference is that the matrices $A_j=WY_j$ are positive-definite now,
and thus can be diagonalized by $\tilde{P}_j\in \mathring{U}(2)$ except
$S_j\in \mathring{U}(1,1)$, and hence $\tilde{V}_j=\tilde{P}_j\tilde{P}_1^*\in \mathring{U}(2)$.
\end{proof}

\section{Preliminary results}

Choose
\begin{equation}\label{U*}
U_*=2|\Lambda| \cdot\Re \left(a_+^2/2- i\lambda_0 a_+-
\log a_+\right).
\end{equation}
\begin{lemma}\label{l:kont_b}
Let $b_{j,1}\in \mathcal{L}_+(\lambda_0)$, $b_{j,2}\in \mathcal{L}_-(\lambda_0)$, $j\in\Lambda$
and let $T_j\in \mathring{U}(1,1)$, $\hat{B}_j=\hbox{diag}\,\{b_{j,1}, -b_{j,2}\}$. Suppose also
that $|\lambda_0|<\sqrt{2}$.
Then $$\Re\,L_{m}(T,\hat{B})\ge 0,$$ where $L_m(T,\hat{B})$ is defined in (\ref{L}),
and the equality holds if and only if
$\hat{B}_j=L_{\pm}$, $T_j=I$ for each $j\in\Lambda$.
\end{lemma}
\begin{proof}
Rewrite (\ref{L}) as
\begin{multline}\label{L_rew}
L_m(T,\hat{B})=B_+(b_1,\lambda_0)+B_-(b_2,\lambda_0)\\+\alpha\sum\limits_{j\sim j^\prime}
|(T_{j^\prime}T_j^{-1})_{12}|^2(b_{j,1}+b_{j,2})(b_{j^\prime,1}+b_{j^\prime,2}),
\end{multline}
where for $b\in \mathbb{C}^{|\Lambda|}$
\begin{align}\label{B_pm}
B_{\pm}(b,\lambda_0)&=-\alpha\,(\nabla b)^2/2+
\sum\limits_{j\in\Lambda}\big(b_j^2/2\mp i\lambda_0 b_j-
\log b_j-b_\pm\big),\\ \notag
 b_\pm&=a_\pm^2/2- i\lambda_0 a_\pm-
\log a_\pm.
\end{align}
Since $b_{j,1}\in \mathcal{L}_+(\lambda_0)$, $b_{j,2}\in
\mathcal{L}_-(\lambda_0)$, we have
\[
b_{j,1}=r_{j,1}e^{i\phi_+},\quad b_{j,2}=r_{j,2}e^{-i\phi_+},\quad r_{j,1},r_{j,2}\ge 0
\]
with $\phi_+=\hbox{arg}\,a_+$, and thus
\begin{align*}
\Re (b_{j,1}+b_{j,2})(b_{j^\prime,1}&+b_{j^\prime,2})=\dfrac{4-\lambda_0^2}{4}(r_{j,1}+r_{j,2})
(r_{j^\prime,1}+r_{j^\prime,2})\\
&-\dfrac{\lambda_0^2}{4}(r_{j,1}-r_{j,2})(r_{j^\prime,1}-r_{j^\prime,2})\\
=&(r_{j,1}+r_{j,2})(r_{j^\prime,1}+r_{j^\prime,2})-
\dfrac{\lambda_0^2}{2}(r_{j,1}r_{j^\prime,1}+r_{j,2}r_{j^\prime,2})\ge 0
\end{align*}
for $|\lambda_0|<\sqrt{2}$. Hence,
\begin{align}\label{in}
\Re L_m(T,\hat{B})\ge \Re B_+(b_1,\lambda_0)+\Re B_-(b_2,\lambda_0)
\end{align}
and the equality holds only if $(T_{j^\prime}T_j^{-1})_{12}=0$ for all $j\sim j^\prime$.
Since $T_1=I$ this means that $T_j=I$ for each $j\in\Lambda$.

Rewrite (\ref{B_pm}) as
\begin{equation}\label{expr_2}
\begin{array}{c}
\Re\,B_\pm(b,\lambda_0)=\Re\,B_\pm(r\,e^{\pm i\phi_+},\lambda_0)=\sum\limits_{j\in\Lambda} \Big(\dfrac{\lambda_0^2}{2}r_j-
\log r_j -b_\pm\Big)\\
-\dfrac{\cos 2\phi_+}{2}
\cdot\Big(\alpha\sum\limits_{j\sim j^\prime}
(r_j-r_{j^\prime})^2-\sum\limits_{j\in\Lambda}r_j^2\Big)\\
=-\dfrac{2-\lambda_0^2}{4}\left(\alpha\sum\limits_{j\sim j^\prime}
(r_j-r_{j^\prime})^2-\sum\limits_{j\in\Lambda}(r_j-1)^2\right)
+\sum\limits_{j\in\Lambda} (r_j- \log r_j-1).
\end{array}
\end{equation}
Since $I+\alpha\Delta>C>0$ for any $\alpha<1/4d$, we have
\begin{equation}\label{kv_form_1}
-\alpha\sum\limits_{j\sim j^\prime}(x_j-x_{j^\prime})^2
+\sum\limits_{j\in\Lambda} x_j^2\ge C\sum\limits_{j\in\Lambda} x_j^2.
\end{equation}
and the equality holds only at $x_1=\ldots=x_n=0$. Besides, for $r_j\ge 0$
\begin{equation}\label{pos}
r_j-\log r_j-1\ge 0
\end{equation}
and the equality holds only for $r_j=1$. This, (\ref{kv_form_1}) for $x_j=r_j-1$, and (\ref{in})  prove
Lemma~\ref{l:kont_b}.

\end{proof}
We need the analogous lemma for the function $K_{m}(V,\hat{U})$ of (\ref{K}):
\begin{lemma}\label{l:kont_u}
Let $u_{j,1}, u_{j,2}\in \mathbb{T}$, $j\in\Lambda$,
and let $\hat{U}_j=\mathrm{diag}\,\{u_{j,1}, u_{j,2}\}$, $V_j\in \mathring{U}(2)$.
Then $$\Re\,K_{m}(V,\hat{U})\ge 0,$$ where $K_{m}(V,\hat{U})$ is defined in (\ref{K}),
and the equality holds if and only if one of the following conditions holds
\begin{enumerate}
    \item $\hat{U}_j=L_{\pm}\,$ or $\,\hat{U}_j=L_{\mp}$, $j\in\Lambda$, and
    $$
    |(V_j)_{12}|=\left\{\begin{array}{ll}
    0,& \hat{U}_j=\hat{U}_1,\\
    1,& \mathrm{otherwise}.
    \end{array}\right.
    $$

    \item $\hat{U}_j=L_{+}$, $j\in\Lambda$.

    \item $\hat{U}_j=L_{-}$, $j\in\Lambda$.
\end{enumerate}

\end{lemma}
\begin{proof}
Rewrite (\ref{K}) as
\begin{multline*}
K_{m}(V,\hat{U})=F(u_{1},\lambda_0)+F(u_{2},\lambda_0)\\+
\alpha\sum\limits_{j\sim j^\prime} |(V_{j^\prime}V_j^{-1})_{12}|^2(u_{j,1}-u_{j,2})(u_{j^\prime,1}-
u_{j^\prime,2}),
\end{multline*}
where for $u\in \mathbb{T}^{|\Lambda|}$
\begin{equation*}
F(u,\lambda_0)=\alpha\,(\nabla u)^2/2-
\sum\limits_{j\in\Lambda}\Big(u_j^2/2
-i\lambda_0 u_j-
\log u_j-b_\pm\Big),
\end{equation*}
and $b_\pm$ is defined in (\ref{B_pm}).

Set
\[
\Delta_{jj^\prime}=
\left\{\begin{array}{ll}
0,& \Re (u_{j,1}-u_{j,2})(u_{j^\prime,1}-u_{j^\prime,2})\ge 0,\\
1,& \Re (u_{j,1}-u_{j,2})(u_{j^\prime,1}-u_{j^\prime,2})<0,\\
%\hbox{any},& \Re (u_{j,1}-u_{j,2})(u_{j^\prime,1}-u_{j^\prime,2})=0,
\end{array}\right.
\]
Then
\begin{multline}\label{Re_in}
\Re K_{m}(V,\hat{U})\ge \Re F(u_{1},\lambda_0)+\Re F(u_{2},\lambda_0)\\+
\alpha\sum\limits_{j\sim j^\prime} \Delta_{jj^\prime} \cdot \Re (u_{j,1}-u_{j,2})(u_{j^\prime,1}-
u_{j^\prime,2}),
\end{multline}
end the equality holds if and only if
\begin{equation*}
|(V_{j^\prime}V_j^{-1})_{12}|^2=
\left\{\begin{array}{ll}
\Delta_{jj^\prime},& \Re (u_{j,1}-u_{j,2})(u_{j^\prime,1}-u_{j^\prime,2})\ne 0,\\
\hbox{any},& \Re (u_{j,1}-u_{j,2})(u_{j^\prime,1}-u_{j^\prime,2})=0.\\
\end{array}\right.
\end{equation*}
The expression
\begin{equation}\label{lapl}
\dfrac{1}{2}\sum\limits_{j\sim j^\prime} (u_{j,1}-u_{j^\prime,1})^2+\dfrac{1}{2}
\sum\limits_{j\sim j^\prime} (u_{j,2}-u_{j^\prime,2})^2
\end{equation}
is a discrete Laplacian operator on the graph $G$ which consists of
two connected components $\Lambda^{(1)}$ (it corresponds to $u_{j,1}$) and $\Lambda^{(2)}$
(it corresponds to $u_{j,2}$), which are two copies of the box $\Lambda$.

Note that
\begin{multline*}
(u_{j,1}-u_{j^\prime,1})^2/2+(u_{j,2}-u_{j^\prime,2})^2/2+
(u_{j,1}-u_{j,2})(u_{j^\prime,1}-u_{j^\prime,2})\\
=(u_{j,2}-u_{j^\prime,1})^2/2+
(u_{j,1}-u_{j^\prime,2})^2/2,
\end{multline*}
and thus adding $(u_{j,1}-u_{j,2})(u_{j^\prime,1}-u_{j^\prime,2})$ to (\ref{lapl})
we change edges $$(\{1,j\},\{1,j^\prime\}) \,\,\hbox{and}\,\,(\{2,j\},\{2,j^\prime\})$$ in the graph $G$
to edges
$$(\{1,j\}, \{2,j^\prime\}) \,\,\hbox{and}\,\,(\{2,j\}, \{1,j^\prime\}).$$

Hence, the r.h.s. of (\ref{Re_in}) is equal to
\begin{multline}\label{graph}
\Re F_{\tilde{G}}(u,\lambda_0):=\alpha\,\Re \sum\limits_{j\sim j^\prime\in \tilde{G}}
(u_j-u_{j^\prime})^2/2\\-
\Re\sum\limits_{j\in\tilde{G}}\Big(u_j^2/2
-i\lambda_0 u_j-
\log u_j-b_\pm\Big),
\end{multline}
where $\tilde{G}$ is a graph obtained from $G$ after all edges changes, which correspond to the
pair $j\sim j^\prime$ such that $\Delta_{jj^\prime}=1$.

Consider $u_j=e^{i\phi_j}\in \mathbb{T}$, $j\in \tilde{G}$ and write
\begin{equation}\label{expr_1}
\begin{array}{c}
\Re\,F_{\tilde{G}}(u,\lambda_0)=\dfrac{\alpha}{2}\sum\limits_{j\sim j^\prime\in \tilde{G}}
\Big((\cos\phi_j-\cos\phi_{j^\prime})^2-
(\sin\phi_j-\sin\phi_{j^\prime})^2\Big)\\
-\sum\limits_{j\in \tilde{G}}\left(
(\cos^2\phi_j-\sin^2\phi_j)/2+\lambda_0\sin\phi_j-\dfrac{2+\lambda_0^2}{4}\right)\\
=\dfrac{\alpha}{2}\sum\limits_{j\sim j^\prime\in \tilde{G}}\Big((\cos\phi_j-\cos\phi_{j^\prime})^2-
(\sin\phi_j-\sin\phi_{j^\prime})^2\Big)\\
+\sum\limits_{j\in\tilde{G}}(\sin\phi_j-\lambda_0/2)^2.
\end{array}
\end{equation}
Since $I+\alpha\Delta>C>0$ for any $\alpha<1/4d$ on any finite graph of degree $2d$, we have
\begin{equation*}
-\dfrac{\alpha}{2}\sum\limits_{j\sim j^\prime\in \tilde{G}}(x_j-x_{j^\prime})^2
+\sum\limits_{j\in \tilde{G}} x_j^2\ge C\sum\limits_{j\in\Lambda}x_j^2
\end{equation*}
and the equality holds only if $x_1=\ldots=x_n=0$. Using this for $x_j=\sin\phi_j-\lambda_0/2$,
we get that (\ref{expr_1}) (and thus the r.h.s. of (\ref{Re_in})) is non-negative and it is zero if
and only if
\begin{align*}
&\sin\phi_j=\lambda_0/2,\quad j\in\tilde{G},\\
&\cos\phi_j=\hbox{const}\,\, \hbox{for\,\,each\,\,connected\,\,component\,\,of}\,\, \tilde{G},
\end{align*}
which means that $u_{j,1}$ and $u_{j,2}$ are equal to $a_\pm$ and are the same for
each connected component of $\tilde{G}$. Thus, $\Re\,K_{m}(V,\hat{U})\ge 0$ and at the minimum point $$u_{j,1}-u_{j,2}\in \mathbb{R}.$$
Taking into account that the Haar measure over $\mathring{U}(2)$ is invariant with respect to the shifting
\[
P_j\to \left(\begin{array}{ll}
0&1\\
1&0
\end{array}
\right)P_j,
\]
where $P_j$ is defined in (\ref{diag}), we can assume without loss of generality that at the minimum point
$$u_{j,1}-u_{j,2}\ge 0.$$
Then $\Delta_{jj^\prime}=0$, thus $\tilde{G}=G$, and hence (\ref{expr_1}) is zero if and only
if fields $u_{j,1}$ and $u_{j,2}$ are constant, and this constants equal to $a_{\pm}$.
This gives the assertion of the lemma.

%Moreover, if $|\phi_j-\phi_{\pm}|\ge \delta$, then $|\sin\phi_j-\lambda_0|\le C\delta$, and hence
%(\ref{kv_form}) gives
%\[
%\dfrac{\alpha}{2}\sum\limits_{j=2}^m(\cos\phi_j-\cos\phi_{j^\prime})^2-
%\dfrac{\alpha}{2}\sum\limits_{j=2}^m(\sin\phi_j-\sin\phi_{j^\prime})^2
%+\sum\limits_{j\in\Lambda}(\sin\phi_j-\lambda_0/2)^2\ge C\delta^2.
%\]
%Thus,
%\begin{equation*}
%\intd_{|\phi_j-\phi_\pm|>\delta} \exp\{-WF(e^{i\phi_j},\lambda_0)\} d\phi
%%{\intd_{0}^{2\pi} \exp\{-WF(e^{i\phi_j},\lambda_0)\} d\phi}
%\le (2\pi)^m\cdot e^{-CW\delta^2}\le e^{-C_1W\delta^2}
%\end{equation*}
%for $\delta=W^{-\kappa}$, where $\kappa<\theta/2$. Also it easy to see that if $\phi_j\in U_\delta(\phi_+)$
%and $\phi_{j^\prime}\in U_\delta(\phi_-)$, then $|\cos\phi_j-\cos\phi_{j^\prime}|>C>0$ and hence
%$\Re\,F(e^{i\phi_j},\lambda_0)<-C$, which proves the second assertion of the lemma.
\end{proof}

\begin{lemma}\label{l:kont_b_+}
Let $a_{j,1}, a_{j,2}\in \mathcal{L}_+(\lambda_0)$, $j\in\Lambda$
and let $\hat{A}_j=\hbox{diag}\,\{a_{j,1}, a_{j,2}\}$, $\tilde{V}_j\in \mathring{U}(2)$. Suppose also
that $|\lambda_0|<\sqrt{2}$.
Then $$\Re\,\tilde{L}_{m}(\tilde{V},\hat{A})\ge 0,$$ where $\tilde{L}_m(\tilde{V},\hat{A})$
is defined in (\ref{L_til}), and the equality holds if and only if
$\hat{A}_j=L_{+}$ for each $j\in\Lambda$.
\end{lemma}
\begin{proof}
Putting
\[
a_{j,1}=r_{j,1}e^{i\phi_+},\quad a_{j,2}=r_{j,2}e^{i\phi_+},\quad r_{j,1},r_{j,2}\ge 0
\]
and rewriting $\Re\,\tilde{L}_{m}(\tilde{V},\hat{A})$, we get
\begin{multline*}
\Re\,\tilde{L}_{m}(\tilde{V},\hat{A})=\tilde{B}(r_{1},\lambda_0)+\tilde{B}(r_{2},\lambda_0)\\-
\dfrac{\alpha (2-\lambda_0^2)}{2}\sum\limits_{j\sim j^\prime} |(V_{j^\prime}V_j^{-1})_{12}|^2(r_{j,1}-r_{j,2})(r_{j^\prime,1}-
r_{j^\prime,2}),
\end{multline*}
where
\[
\tilde{B}(r,\lambda_0)=-\dfrac{\alpha(2-\lambda_0^2)}{4}(\nabla r)^2+
\sum\limits_{j\in\Lambda}\big((2-\lambda_0^2)r_j^2/4+\lambda_0^2 r_j/2-
\log r_j-b_+\big),
\]
and $b_+$ is defined in (\ref{B_pm}). Similarly to (\ref{Re_in}) -- (\ref{graph}), we get
\begin{align*}
\Re\,\tilde{L}_{m}(\tilde{V},\hat{A})&\ge -
\dfrac{\alpha(2-\lambda_0^2)}{4} \sum\limits_{j\sim j^\prime\in \tilde{G}} (r_j-r_{j^\prime})^2/2\\
&+\sum\limits_{j\in\tilde{G}}\big((2-\lambda_0^2)r_j^2/4+\lambda_0^2 r_j/2-
\log r_j-b_+\big)\\
&=\dfrac{2-\lambda_0^2}{2}\Big(-\alpha \sum\limits_{j\sim j^\prime\in \tilde{G}} (r_j-r_{j^\prime})^2/2
+\sum\limits_{j\in\tilde{G}}(r_j-1)^2/2\big)\\
&+\sum\limits_{j\in\tilde{G}}(r_j-\log r_j-1),
\end{align*}
where $\tilde{G}$ is a graph obtained from two copies of $\Lambda$ by all changes which correspond
to $$(r_{j,1}-r_{j,2})(r_{j^\prime,1}-
r_{j^\prime,2})>0.$$
This and (\ref{pos}) yield the lemma.
\end{proof}

\section{Proof of Theorem \ref{thm:1}}
According to (\ref{main}) -- (\ref{F_expr}), we can rewrite Theorem \ref{thm:1} as
\begin{equation}\label{reform}
\lim\limits_{\varepsilon\to 0}\,\lim\limits_{W\to\infty} F_2(z_1,z_2)
=1-\dfrac{\sin^2 (\pi(\xi_1-\xi_2))}
{\pi^2(\xi_1-\xi_2)^2}.
\end{equation}
The proof of (\ref{reform}) can be divided into two theorems:
\begin{theorem}\label{thm:+-}
We have for $G_2^{+-}(z,\xi)$ of (\ref{G_2})
\begin{multline}\label{+-}
\lim\limits_{\varepsilon\to 0}\,\lim\limits_{W\to\infty}\dfrac{\partial^2}
{\partial\xi_1^\prime\partial\xi_2^\prime}\left(G_2^{+-}(z,\xi)+\overline{G}_2^{\,+-}(z,\xi)\right)
\Big|_{\xi^\prime=\xi}\\
=-\dfrac{2}{\rho^2(\lambda_0)}+\dfrac{4\sin^2(\pi(\xi_1-\xi_2))}{(\xi_1-\xi_2)^2},
\end{multline}
where $\xi^\prime=\xi$ means $\xi_1^\prime=\xi_1$, $\xi_2^\prime=\xi_2$.
\end{theorem}
\begin{theorem}\label{thm:++}
We have for $G_2^{++}(z,\xi)$ of (\ref{G_2})
\begin{equation}\label{++}
\lim\limits_{\varepsilon\to 0}\,\lim\limits_{W\to\infty}\dfrac{\partial^2}
{\partial\xi_1^\prime\partial\xi_2^\prime}\left(G_2^{++}(z,\xi)+\overline{G}_2^{\,++}(z,\xi)\right)
=\dfrac{a_+^2+a_-^2}{\rho^2(\lambda_0)}.
\end{equation}
\end{theorem}
Indeed, (\ref{F_expr}) and (\ref{+-}) -- (\ref{++}) give
\begin{multline*}
\lim\limits_{\varepsilon\to 0}\,\lim\limits_{W\to\infty} F_2(z_1,z_2)\\=
\dfrac{a_+^2+a_-^2+2}{4\pi^2\rho^2(\lambda_0)}-\dfrac{\sin^2(\pi(\xi_1-\xi_2))}{\pi^2(\xi_1-\xi_2)^2}
=1-\dfrac{\sin^2(\pi(\xi_1-\xi_2))}{\pi^2(\xi_1-\xi_2)^2},
\end{multline*}
which yields (\ref{reform}), thus Theorem \ref{thm:1}.

\subsection{Proof of Theorem \ref{thm:+-}}

According to Lemmas \ref{l:kont_b} -- \ref{l:kont_u}, the saddle-points of (\ref{G_int}) can be divided
in the following three subsets:
\begin{description}
    \item[\textbf{type I.}] for each $j\in\Lambda$: $\hat{U}_j=L_{\pm}\,$ or $\,\hat{U}_j=L_{\mp}$,  and
    $$
    |(V_j)_{12}|=\left\{\begin{array}{ll}
    0,& \hat{U}_j=\hat{U}_1,\\
    1,& \mathrm{otherwise},
    \end{array}\right.
    $$
    $\hat{B}_j=L_\pm$, $T_j=I$;
    \item[\textbf{type II.}] for each $j\in\Lambda$: $\hat{U}_j=L_+$, $V_j\in\mathring{U}(2)$, $\hat{B}_j=L_\pm$, $T_j=I$;

    \item[\textbf{type III.}] for each $j\in\Lambda$: $\hat{U}_j=L_-$, $V_j\in\mathring{U}(2)$,
    $\hat{B}_j=L_\pm$, $T_j=I$,
\end{description}
where $L_\pm$, $L_+$ and $L_-$ are defined in (\ref{L_pm}).

Set
\begin{align}\notag
\exp\{f_u(V,\hat{U},\hat{\xi})\}&:=\int
e^{-|\Lambda|^{-1}\sum_{j\in\Lambda}\Tr (V_jP_1^*)^*\hat{U}_j(V_jP_1^*)
(i\hat{\xi}/\rho(\lambda_0)-\varepsilon L)}d\mu(P_1),\\ \notag
\exp\{f_b(T,\hat{B},\hat{\xi})\}&:=\int
e^{|\Lambda|^{-1}\sum_{j\in\Lambda}\Tr (T_jS_1^{-1})^{-1}
\hat{B}_j(T_jS_1^{-1})
(i\hat{\xi}/\rho(\lambda_0)-\varepsilon L)}d\nu(S_1),\\ \label{f_sm}
f_u^{(1)}(V,\hat{U},\hat{\xi})&:=\dfrac{\partial}{\partial \xi_1^\prime}
f_u(V,\hat{U},\hat{\xi}_1)\Big|_{\xi_1^\prime=\xi_1},\\ \notag
f_b^{(1)}(T,\hat{B},\hat{\xi})&:=\dfrac{\partial}{\partial \xi_2^\prime}
f_u(T,\hat{B},\hat{\xi}_2)\Big|_{\xi_2^\prime=\xi_2},
\end{align}
where $\hat{\xi}_1$, $\hat{\xi}_2$ are defined in (\ref{xi_hat}).

Introduce
\begin{align}\label{mes}
\langle  F\rangle_{\delta, k}&=
\dfrac{W^{4|\Lambda|}}{(8\pi^2)^{|\Lambda|}}\displaystyle\int \prod\limits_{j\in\Lambda\setminus\{\overline{1}\}} d\mu(V_j)d \nu(T_j)
\int\limits_{U_{\delta,k}} \prod\limits_{j\in\Lambda} d u_{j,1}du_{j,2}
d b_{j,1}db_{j,2}\\ \notag
&\times F\cdot \exp\Big\{-W\left(K_{m}(V,\hat{U})+L_{m}(T,\hat{B})\right)\Big\}\cdot
\prod\limits_{j\in\Lambda}
(b_{j,1}+b_{j,2})^2\\ \notag
&\times \int d\mu(P_1)d \nu(S_1) \exp\Big\{\mathcal{F}_m(P_1,S_1,V,T,\hat{U},\hat{B})\Big\}
\cdot\prod\limits_{j\in\Lambda}
(u_{j,1}-u_{j,2})^2.
\end{align}
where $U_{\delta,k}$ is a $\delta$-neighborhood of the saddle-point $k$, and let
also $\langle \ldots \rangle_\delta$ be the sum of all $\langle\ldots\rangle_{\delta,k}$,
and $\delta$ is defined in (\ref{delta}).

Then Lemmas \ref{l:kont_b} -- \ref{l:kont_u} yield
\begin{align}\notag
G_2^{+-}(z,\xi)\Big|_{\xi^\prime=\xi}&=\langle \mathcal{P}_N\rangle_\delta+o(1),\\ \label{okr}
\dfrac{\partial^2}{\partial \xi_1^\prime\partial \xi_2^\prime}G^{+-}_2(z,\xi)
\Big|_{\xi=\xi^\prime}&=\langle f^{(1)}_u\cdot f^{(1)}_b\cdot \mathcal{P}_N\rangle_\delta+o(1),\\ \notag
\dfrac{\partial}{\partial \xi_1^\prime}G_2^{+-}(z,\xi)\Big|_{\xi^\prime=\xi}&=
\langle  f^{(1)}_u\cdot \mathcal{P}_N\rangle_\delta+o(1),\\ \notag
\dfrac{\partial}{\partial \xi_2^\prime}G_2^{+-}(z,\xi)\Big|_{\xi^\prime=\xi}&=\langle
f^{(1)}_b\cdot \mathcal{P}_N\rangle_\delta+o(1)
\end{align}
with $\delta=\log W/W^{1/2}$. To simplify the formulas below, we will omit the index $\delta$ in further calculations.

%According to Lemmas \ref{l:kont_b} -- \ref{l:kont_u} saddle-points of (\ref{G_int})
%belong to the following three types:
%\begin{description}
%    \item[type I.] for each $j\in\Lambda$: $\hat{U}_j=L_{\pm}\,$ or $\,\hat{U}_j=L_{\mp}$,  and
%    $$
%    |(V_j)_{12}|=\left\{\begin{array}{ll}
%    0,& \hat{U}_j=\hat{U}_1,\\
%    1,& \mathrm{otherwise},
%    \end{array}\right.
%    $$
%    $\hat{B}_j=L_\pm$, $T_j=I$;
%    \item[type II.] for each $j\in\Lambda$: $\hat{U}_j=L_+$, $\hat{B}_j=L_\pm$, $T_j=I$;
%
%    \item[type III.] for each $j\in\Lambda$: $\hat{U}_j=L_-$, $\hat{B}_j=L_\pm$, $T_j=I$.
%\end{description}

%Indeed, consider the contour
%\begin{multline}\label{C_R}
%\mathcal{C}^{\pm}_R=[0,R]\cup\{z\in \mathbb{C}: z=t\cdot e^{i\phi_\pm},\,t\in [0,R/\cos\phi_\pm]\}\\
%\cup \{z\in \mathbb{C}: z=R\pm it,\,t\in [0,R\tan\phi_+]\}
%\end{multline}
%for some $R>0$. Since integrand in (\ref{G_int}) is analytic with respect to $\{b_{j,1}\}$, $\{b_{j,2}\}$,
%if we change one integration in $b$ over $\mathbb{R}_+$ to the integration over $\mathcal{C}^{\pm}_R$
%we get $0$. Hence we are left to prove that the integral (\ref{G_int}) with ........

Note also that the contributions of all points of type I are the same since
one of such points can be obtained from another one by the rotation
\[
P_j\to \left(\begin{array}{ll}
0&1\\
1&0
\end{array}
\right)P_j
\]
for some $j$, where $P_j$ is defined in (\ref{diag}). Hence,
\begin{equation*}
\langle f^{(1)}_u\cdot f^{(1)}_b\cdot \mathcal{P}_N\rangle_I=2^{|\Lambda|}
\langle f^{(1)}_u\cdot f^{(1)}_b\cdot \mathcal{P}_N\rangle_1,
\end{equation*}
where $\langle\ldots\rangle_1$ is $\langle\ldots\rangle_k$ for the saddle-point
$\hat{U_j}=\hat{B}_j=L_{\pm}$, $V_j=T_j=I$ (we will denote this saddle-point $1$).

Take now the $\delta$-neighborhood of one saddle-point $\hat{U_j}=L_s$ ($L_s$ can be equal to
$L_\pm$, $L_+$ or $L_-$), $\hat{B}_j=L_{\pm}$, $T_j=I$ and change variables as
\begin{align*}
&u_{j,1}=a_{s_1}\exp\{i\tilde{u}_{j,1}/\sqrt{W}\},\quad u_{j,2}=a_{s_2}\exp\{i\tilde{u}_{j,2}/\sqrt{W}\}, \\
&b_{j,1}=a_+(1+\tilde{b}_{j,1}/\sqrt{W}),\quad \quad b_{j,2}=-a_-(1+\tilde{b}_{j,2}/\sqrt{W}),\\
&t_j=\tilde{t}_j/\sqrt{W},
\end{align*}
where $L_s=\hbox{diag}\,\{a_{s_1}, a_{s_2}\}$, $|\tilde{u}_{j,1}|, |\tilde{u}_{j,2}|\le \log W $,
$|\tilde{b}_{j,1}|, |\tilde{b}_{j,2}|\le \log W$, and $\tilde{t}_j \in [0,\log W]$.
This change is chosen to kill the big parameter $W$ in front of $K_m$, $L_m$.

Then
\begin{align}\label{repr_okr}
\hat{U}_j&=L_s+iL_s\cdot\hbox{diag}\,\{\tilde{u}_{j,1},\tilde{u}_{j,2}\}/\sqrt{W}+R_j/W,\\ \notag
\hat{B}_j&=L_\pm+L_\pm\cdot\hbox{diag}\,\{\tilde{b}_{j,1},\tilde{b}_{j,2}\}/\sqrt{W},\\ \notag
T_j&=I+\dfrac{\tilde{t}_j}{\sqrt{W}}\left(\begin{array}{cc} 0&e^{i\sigma_j}\\
e^{-i\sigma_j}&0\end{array}\right)+r_jI/W,
\end{align}
where
\begin{align*}
\notag R_j&=L_s\left(-\hbox{diag}\,\{\tilde{u}_{j,1}^2,
\tilde{u}^2_{j,2}\}/2-i\cdot\hbox{diag}\,\{\tilde{u}_{j,1}^3,
\tilde{u}^3_{j,2}\}/6\sqrt{W}+\ldots\right)=O(\log^2 W),\\
r_j&=W((1+\tilde{t}_j^{\phantom{\,\,} 2}/W)^{1/2}-1)=O(\log^2 W).
\end{align*}
If $L_s=L_\pm$, then we also change $v_j=\tilde{v}_j/\sqrt{W}$, $\tilde{v}_j\in [0,\log W]$, which gives
\begin{align}\label{V}
V_j&=I+\dfrac{\tilde{v}_j}{\sqrt{W}}\left(\begin{array}{cc} 0&e^{i\theta_j}\\
e^{-i\theta_j}&0\end{array}\right)+p_jI/W,\\ \notag
p_j&=W((1+\tilde{v}_j^2/W)^{1/2}-1)=O(\log^2 W).
\end{align}
Substituting this to $K_m$, $L_m$ of (\ref{K}) -- (\ref{L}), we get
\begin{align*}
&WK_{m}(V,\hat{U})=K_m^{(s)}(V,\hat{U})+
O(\log^3 W/\sqrt{W})\\
&WL_{m}(T,\hat{B})=L_m^{(0)}(\tilde{t},\tilde{B})+
O(\log^3 W/\sqrt{W}),
\end{align*}
where
\begin{align}\label{KL_s}
&K_m^{(s)}(V,\hat{U})=(M_{s_1}\tilde{u}_1,\tilde{u}_1)/2+
(M_{s_2}\tilde{u}_2,\tilde{u}_2)/2
+\alpha W \sum\limits_{j\sim j^\prime}
|(V_{j^\prime}V_j^*)_{12}|^2\\ \notag
&\times\Big(a_{s_1}-a_{s_2}+\dfrac{i(a_{s_1}\tilde{u}_{j,1}-a_{s_2}\tilde{u}_{j,2})}{\sqrt{W}}\Big)
\Big(a_{s_1}-a_{s_2}+\dfrac{i(a_{s_1}\tilde{u}_{j^\prime,1}-a_{s_2}\tilde{u}_{j^\prime,2})}{\sqrt{W}}\Big)\\ \notag
&L_m^{(0)}(\tilde{t},\tilde{B})=(M_+\tilde{b}_1,\tilde{b}_1)/2+
(M_-\tilde{b}_2,\tilde{b}_2)/2\\ \notag
&\quad\quad\quad\quad\quad+\alpha (a_+-a_-)^2\sum\limits_{j\sim j^\prime}
|\tilde{t}_je^{i\sigma_j}-\tilde{t}_{j^\prime}e^{i\sigma_{j^\prime}}|^2,
\end{align}
and $M_\pm=\alpha a_\pm^2\Delta+(1+a_\pm^2)I$. Here $s_1$, $s_2$ equal to $+$ or $-$,
$L_s=\hbox{diag}\,\{a_{s_1},a_{s_2}\}$.

Hence, we have from (\ref{mes})
\begin{align}\notag
&\big\langle f^{(1)}_u\cdot f^{(1)}_b\cdot \mathcal{P}_N \big\rangle_{k}=
-\dfrac{W^{|\Lambda|+1}}{(8\pi^2)^{|\Lambda|}}\displaystyle\int \prod\limits_{j\in\Lambda\setminus\{\overline{1}\}}
d\mu(V_j)
\int\limits_{-\log W}^{\log W} d\tilde{\vphantom{b}u}d\tilde{b}\\ \notag
&\times\int\limits_{0}^{\log W}\prod\limits_{j\in\Lambda\setminus\{\overline{1}\}}(2\tilde{t}_jd\tilde{t}_j)
\int\limits_{0}^{2\pi}\prod\limits_{j\in\Lambda\setminus\{\overline{1}\}}\dfrac{d\sigma_j}{2\pi}\cdot
\prod\limits_{j\in\Lambda}\big(a_{s_1}a_{s_2}\cdot e^{\frac{i(\tilde{u}_{j,1}+
\tilde{u}_{j,2})}{\sqrt{W}}}\big)\\ \label{int_okr}
&\times f^{(1)}_u(V,\hat{U},\hat{\xi})\cdot f^{(1)}_b(T,\hat{B},\hat{\xi})\cdot
\mathcal{P}_N(V^*\hat{U}V,T^{-1}\hat{B}T)\\ \notag
&\times
\exp\Big\{-K_{m}^{(s)}(V,\hat{U})-L_{m}^{(0)}(\tilde{t},\tilde{B})+O(\log^3W/\sqrt{W})\Big\}
 \\ \notag
&\times \exp\Big\{f_u(V,\hat{U},\hat{\xi})+f_b(T,\hat{B},\hat{\xi})\Big\}\cdot \prod\limits_{j\in\Lambda}
\Big (a_+-a_-+\dfrac{a_+\tilde{b}_{j,1}+a_-\tilde{b}_{j,2}}{\sqrt{W}}\Big)^2\\ \notag
&\times\prod\limits_{j\in\Lambda}
\Big(a_{s_1}-a_{s_2}+\dfrac{ia_{s_1}\tilde{u}_{j,1}-ia_{s_2}\tilde{u}_{j,2}}{\sqrt{W}}+O(\log^2W/W)\Big)^2
+o(1),
\end{align}
where
\[
d\tilde{u}=\prod\limits_{j\in\Lambda} d\tilde{u}_{j,1}d\tilde{u}_{j,2},\quad
d\tilde{b}=\prod\limits_{j\in\Lambda} d\tilde{b}_{j,1}d\tilde{b}_{j,2},
\]
$k$ is the saddle-point
$\hat{U}_j=L_s$, $\hat{B}_j=L_\pm$, $T_j=I$, and $V_j=I$ for $s=\pm$ and $V_j\in \mathring{U}(2)$
for $s\ne \pm$.

Moreover, (\ref{P_n}) yields
\begin{align}\label{P_n_1}
\mathcal{P}_N(U,B)&=\intd\exp\Big\{\alpha \sum\limits_{j\sim j^\prime}
\Tr (\rho_j-\rho_{j^\prime})
(\tau_j-\tau_{j^\prime})-\sum\limits_{j\in\Lambda} \Tr \rho_j\tau_j\Big\}\\ \notag
&\times\exp\Big\{-W\sum\limits_{j\in\Lambda}\log \det(1+W^{-1}U_j^{-1}\rho_jB_j^{-1}\tau_j)\Big\}
\prod\limits_{j\in\Lambda}
d\rho_jd\tau_j\\ \notag
&=\intd\exp\Big\{\alpha \sum\limits_{j\sim j^\prime}
\Tr (\rho_j-\rho_{j^\prime})
(\tau_j-\tau_{j^\prime})-\sum\limits_{j\in\Lambda} \Tr \rho_j\tau_j\Big\}\\ \notag
&\times\exp\Big\{-\sum\limits_{j\in\Lambda}\Tr U_j^{-1}\rho_jB_j^{-1}\tau_j\Big\}
(1+W^{-1}P_N(\tau,\rho, U,B))
\prod\limits_{j\in\Lambda}
d\rho_jd\tau_j\\ \notag
&=\mdet \left(\alpha\Delta+I+\hbox{diag}\{U_j^{-1}\otimes B_j^{-1}\}_{j\in\Lambda}\right)+W^{-1}
\tilde{\mathcal{P}}_N(U,B),
\end{align}
where $P_N$ is a polynomial of finite degree of $\rho_j$, $\tau_j$, $U_j^{-1}$ and
$B_j^{-1}$ with bounded coefficients, and
\begin{align*}
\tilde{\mathcal{P}}_N(U,B)&=\intd\exp\Big\{\alpha \sum\limits_{j\sim j^\prime}
\Tr (\rho_j-\rho_{j^\prime})
(\tau_j-\tau_{j^\prime})-\sum\limits_{j\in\Lambda} \Tr \rho_j\tau_j\Big\}\\ \notag
&\times\exp\Big\{-\sum\limits_{j\in\Lambda}\Tr U_j^{-1}\rho_jB_j^{-1}\tau_j\Big\}
P_N(\tau,\rho, U,B)
\prod\limits_{j\in\Lambda}
d\rho_jd\tau_j.
\end{align*}
Thus, we have in the $\delta$-neighborhood of each saddle-point
\begin{equation}\label{p_cal_til}
|W^{-1}\tilde{\mathcal{P}}_N(U,B)|\le C\log^2 W/W,
\end{equation}
where $\delta$ is defined in (\ref{delta}).
\begin{lemma}\label{l:p_bound}
We have in the $\delta$-neighborhood of each saddle point
\[
|\mathcal{P}_N(V^*\hat{U}V,T^{-1}\hat{B}T)|\le C\log^2 W/W.
\]
\end{lemma}
\begin{proof}
Taking into account (\ref{p_cal_til}), we are left to bound the determinant in the r.h.s. of (\ref{P_n_1}).

Expand the determinant according to (\ref{repr_okr}). Since
\[
\mdet(\alpha\Delta+I+\hbox{diag}\{L_s^{-1}\otimes L_{\pm}^{-1}\}_{j\in\Lambda})=0
\]
for each $L_s$ (because $1+a_+^{-1}a_-^{-1}=0$), it is sufficient to prove
that all first partial derivatives of the determinant with respect to $\tilde{u}$, $\tilde{b}$, $\tilde{v}$ or
$\tilde{t}$ are zero at each saddle-point.
Consider first the case when we differentiate over $\tilde{u}$ or $\tilde{v}$.
Then we can put
\[
T^{-1}_j\hat{B}_j^{-1}T_j=L_{\pm}^{-1},\quad j\in\Lambda,
\]
and differentiate the expression
\begin{multline}\label{det_1}
\mdet(\alpha\Delta+I+\hbox{diag}\{V^*_j\hat{U}_j^{-1}V_j\otimes L_{\pm}^{-1}\}_{j\in\Lambda})\\
=
\mdet(\alpha\Delta+I+a_+^{-1}\hbox{diag}\{V^*_j\hat{U}_j^{-1}V_j\}_{j\in\Lambda})
\\ \times
\mdet(\alpha\Delta+I+a_-^{-1}\hbox{diag}\{V^*_j\hat{U}_j^{-1}V_j\}_{j\in\Lambda}).
\end{multline}
Note that at the saddle-point $1$ both determinant in the r.h.s. are zero and so the first
derivative of the determinant in the l.h.s.
with respect to $\tilde{u}_{j,1}$, $\tilde{u}_{j,2}$ and $\tilde{v}_j$ at the saddle-point $1$ is zero.
Thus, we are left to consider the case $L_s=L_+$ (the case $L_s=L_-$
is similar). Then
\[
V^*_j\hat{U}_j^{-1}V_j=a_+^{-1}(I+V^*_j\tilde{U}_jV_j)^{-1},
\]
where
\[
\tilde{U}_j=\hat{U}_j^{-1}-
L_+^{-1}=O(\log W/\sqrt{W}).
\]
Since for $L_s=L_+$ the second determinant in the r.h.s of (\ref{det_1}) is zero, we must
differentiate
\begin{multline*}
\mdet(\alpha\Delta+I+a_-^{-1}\hbox{diag}\{V^*_j\hat{U}_j^{-1}V_j\}_{j\in\Lambda})\\=
\mdet(\alpha\Delta+I-\hbox{diag}\{(I+V^*_j\tilde{U}_jV_j)^{-1}\}_{j\in\Lambda})\\
=\mdet(\alpha\Delta+\hbox{diag}\{V^*_j\tilde{U}_jV_j/\sqrt{W}\}_{j\in\Lambda})+\hbox{higher\,\,orders}.
\end{multline*}
But it is easy to see that the first derivative of the r.h.s. with respect to $\tilde{u}_{j,1}$,
$\tilde{u}_{j,2}$ or $\tilde{v}_j$ is zero, and so in this case the lemma is also proven.

Let now differentiate with respect to $\tilde{b}_{j,1}$,
$\tilde{b}_{j,2}$ or $\tilde{t}_j$. Similarly to (\ref{det_1})
we need to differentiate
\begin{multline}\label{det_2}
\mdet(\alpha\Delta+I+\hbox{diag}\{L_s\otimes T^{-1}_j\hat{B}_j^{-1}T_j\}_{j\in\Lambda})\\
=
\mdet(\alpha\Delta+I+a_{s_1}^{-1}\hbox{diag}\{T^{-1}_j\hat{B}_j^{-1}T_j\}_{j\in\Lambda})\\ \times
\mdet(\alpha\Delta+I+a_{s_2}^{-1}\hbox{diag}\{T^{-1}_j\hat{B}_j^{-1}T_j\}_{j\in\Lambda}).
\end{multline}
Since at each saddle-point both determinants in the r.h.s. of (\ref{det_2}) are zero, the first
derivatives with respect
to $\tilde{b}_{j,1}$, $\tilde{b}_{j,2}$ and $\tilde{t}_j$ at any saddle-point
are zero, which completes the proof of the lemma.$\quad\Box$
\end{proof}
Set
\begin{align}\label{f*}
f^{(1)}_{u,*}:=f^{(1)}_{u}(I,L_\pm,\hat{\xi}),\quad f^{(1)}_{b,*}:=f^{(1)}_{b}(I,L_\pm,\hat{\xi}).
\end{align}

The next step is to prove
\begin{lemma}\label{l:trick}
In the notations (\ref{mes}) and (\ref{f*}) we have
\begin{align*}
\langle f^{(1)}_u f^{(1)}_b\cdot \mathcal{P}_N\rangle &=
2^{|\Lambda|}\langle (f^{(1)}_u(V,\hat{U},\hat{\xi})-f^{(1)}_{u,*}) (f^{(1)}_b(T,\hat{B},\hat{\xi})-f^{(1)}_{b,*})
\cdot \mathcal{P}_N\rangle_1\\
&+ia_-f^{(1)}_{u,*}/\rho(\lambda_0)
-ia_+f^{(1)}_{b,*}/\rho(\lambda_0)-f^{(1)}_{b,*}f^{(1)}_{u,*}+o(1),
\end{align*}
where $1$ is the saddle-point $\hat{U}_j=\hat{B}_j=L_\pm$, $V_j=T_j=I$.
\end{lemma}
Roughly speaking, in this lemma using (\ref{okr}) we change the order of the first
non-zero coefficient in the expansion of $f^{(1)}_u f^{(1)}_b \mathcal{P}_N$ from 2 to 4,
which helps to omit all saddle-points of type II and III and simplify calculations for the points
of type I (see the beginning of the proof of Theorem \ref{thm:+-}).
\begin{proof}
It follows from (\ref{Ward_id}) and (\ref{okr}) that
\begin{align*}
\langle \mathcal{P}_N\rangle&=1+o(1),\\
\langle f^{(1)}_b\cdot \mathcal{P}_N\rangle&=
\mathbf{E}\left\{(\rho(\lambda_0) N)^{-1}\Tr G(\overline{z}_2)\right\}+o(1)=ia_-/\rho(\lambda_0)+o(1),\\
\langle f^{(1)}_u\cdot \mathcal{P}_N\rangle&=-
\mathbf{E}\left\{(\rho(\lambda_0) N)^{-1}\Tr G(z_1)\right\}+o(1)=-ia_+/\rho(\lambda_0)+o(1),
\end{align*}
Hence,
\begin{align}\label{razl}
\langle f^{(1)}_u f^{(1)}_b\cdot \mathcal{P}_N\rangle &=\langle f^{(1)}_u (f^{(1)}_b(T,\hat{B},\hat{\xi})-f^{(1)}_{b,*})
\cdot \mathcal{P}_N\rangle+f^{(1)}_{b,*}\langle f^{(1)}_u \cdot \mathcal{P}_N\rangle\\ \notag
&=
\langle f^{(1)}_u (f^{(1)}_b(T,\hat{B},\hat{\xi})-f^{(1)}_{b,*})
\cdot \mathcal{P}_N\rangle-ia_+f^{(1)}_{b,*}/\rho(\lambda_0)+o(1),\\ \notag
\langle f^{(1)}_b\cdot \mathcal{P}_N\rangle&=\langle (f^{(1)}_b(T,\hat{B},\hat{\xi})-f^{(1)}_{b,*})
\cdot \mathcal{P}_N\rangle+f^{(1)}_{b,*}+o(1).
\end{align}
We have the coefficient $W^{|\Lambda|+1}$ in front of the integral (\ref{int_okr}).
If $L_s=a_+ I$ or $L_s=a_-I$, then
\[
\Big|\prod\limits_{j\in\Lambda}
\Big(a_{s_1}-a_{s_2}+\dfrac{ia_{s_1}\tilde{u}_{j,1}-ia_{s_2}\tilde{u}_{j,2}}{\sqrt{W}}+O(\log^2W/W)\Big)^2
\Big|\le C\log^{2|\Lambda|}W/W^{|\Lambda|}.
\]
Besides,
\[
|f^{(1)}_b(T,\hat{B},\hat{\xi})-f^{(1)}_{b,*}|\le C\log W/\sqrt{W}.
\]
This and Lemma \ref{l:p_bound} yield
\begin{align*}
&|\langle f^{(1)}_u (f^{(1)}_b(T,\hat{B},\hat{\xi})-f^{(1)}_{b,*})
\cdot \mathcal{P}_N\rangle_{II,III}|\le C\log^{2|\Lambda|+3}W/\sqrt{W}=o(1),\\
& |\langle (f^{(1)}_b(T,\hat{B},\hat{\xi})-f^{(1)}_{b,*})
\cdot \mathcal{P}_N\rangle_{II,III}|\le C\log^{2|\Lambda|+3}W/\sqrt{W}=o(1).
\end{align*}
Thus, we get from (\ref{razl})
\begin{align*}
\langle f^{(1)}_u f^{(1)}_b\cdot \mathcal{P}_N\rangle &=
\langle f^{(1)}_u (f^{(1)}_b(T,\hat{B},\hat{\xi})-f^{(1)}_{b,*})
\cdot \mathcal{P}_N\rangle_I-ia_+f^{(1)}_{b,*}/\rho(\lambda_0)+o(1),\\
\langle f^{(1)}_b\cdot \mathcal{P}_N\rangle&=\langle (f^{(1)}_b(T,\hat{B},\hat{\xi})-f^{(1)}_{b,*})
\cdot \mathcal{P}_N\rangle_I+f^{(1)}_{b,*}+o(1),
\end{align*}
where $\langle\ldots\rangle_I$ is the integral over the union of all $\delta$-neighborhoods
of the points of type I (see the beginning of the proof of Theorem \ref{thm:+-}). As it was mentioned before, the contributions of all such points are
equal and hence we can consider only the contribution of the point $1$.
Hence,
\begin{align*}
\langle f^{(1)}_u f^{(1)}_b\cdot \mathcal{P}_N\rangle &=
2^{|\Lambda|} \langle (f^{(1)}_u(V,\hat{U},\hat{\xi})-f^{(1)}_{u,*}) (f^{(1)}_b(T,\hat{B},\hat{\xi})-f^{(1)}_{b,*})
\cdot \mathcal{P}_N\rangle_1\\
&+f^{(1)}_{u,*}\langle (f^{(1)}_b(T,\hat{B},\hat{\xi})-f^{(1)}_{b,*})
\cdot \mathcal{P}_N\rangle_1
-ia_+f^{(1)}_{b,*}/\rho(\lambda_0)+o(1)\\
&=2^{|\Lambda|}\langle (f^{(1)}_u(V,\hat{U},\hat{\xi})-f^{(1)}_{u,*}) (f^{(1)}_b(T,\hat{B},\hat{\xi})-f^{(1)}_{b,*})
\cdot \mathcal{P}_N\rangle_1\\
&+ia_-f^{(1)}_{u,*}/\rho(\lambda_0)
-ia_+f^{(1)}_{b,*}/\rho(\lambda_0)-f^{(1)}_{b,*}f^{(1)}_{u,*}+o(1),
\end{align*}
and the lemma is proven. $\quad\Box$
\end{proof}
Taking into account Lemma \ref{l:trick}, we have to compute only
\[
\langle (f^{(1)}_u(V,\hat{U},\hat{\xi})-f^{(1)}_{u,*}) (f^{(1)}_b(T,\hat{B},\hat{\xi})-f^{(1)}_{b,*})
\cdot \mathcal{P}_N\rangle_1.
\]
\begin{lemma}\label{l:main_contr}
In the notations (\ref{mes}) and (\ref{f*}) we have
\begin{multline*}
\langle (f^{(1)}_u(V,\hat{U},\hat{\xi})-f^{(1)}_{u,*}) (f^{(1)}_b(T,\hat{B},\hat{\xi})-f^{(1)}_{b,*})
\cdot \mathcal{P}_N\rangle_1\\
=\dfrac{1}{2^{|\Lambda|}(2\theta_\varepsilon\rho(\lambda_0))^2}\cdot\left(-e^{-2ic_0\theta_{\varepsilon}}+
\dfrac{2ic_0\theta_\varepsilon e^{-ic_0\theta_{\varepsilon}}}
{e^{ic_0\theta_{\varepsilon}}-e^{-ic_0\theta_{\varepsilon}}}\right).
\end{multline*}
\end{lemma}
\begin{proof}
Substituting (\ref{V}) to (\ref{KL_s}) with $s=\pm$, we obtain
\begin{align*}
&K_{m}^{(\pm)}(V,\hat{U})=K_{m}^{(0)}(\tilde{v},\tilde{U})+O(\log W/\sqrt{W}),
\end{align*}
where
\begin{multline}\label{K_0}
K_{m}^{(0)}(\tilde{v},\tilde{U})=(M_+\tilde{u}_1,\tilde{u}_1)/2+(M_-\tilde{u}_2,\tilde{u}_2)/2\\
+\alpha (a_+-a_-)^2\sum\limits_{j\sim j^\prime}
|\tilde{v}_je^{i\theta_j}-\tilde{v}_{j^\prime}e^{i\theta_{j^\prime}}|^2.
\end{multline}
Hence,
\begin{align}\label{okr_1}
&\big\langle (f^{(1)}_u-f^{(1)}_{u,*}) (f^{(1)}_b-f^{(1)}_{b,*})
\cdot \mathcal{P}_N\big\rangle_1=
\dfrac{W^2}{(8\pi^2)^{|\Lambda|}}\int\limits_{-\log W}^{\log W} d\tilde{\vphantom{b}u}d\tilde{b}\\ \notag
&\times\displaystyle\int_{0}^{2\pi} \prod\limits_{j\in\Lambda\setminus\{\overline{1}\}}
\dfrac{d\theta_{j}d\sigma_j}{(2\pi)^2}
\int\limits_{0}^{\log W}\prod\limits_{j\in\Lambda\setminus\{\overline{1}\}}(2\tilde{v}_j
d\tilde{v}_j)(2\tilde{t}_j d\tilde{t}_j)
\cdot \prod\limits_{j\in\Lambda}
\ e^{i(\tilde{u}_{j,1}+
\tilde{u}_{j,2})/\sqrt{W}}\\
\notag
&\times (f^{(1)}_u(V,\hat{U},\hat{\xi})-f^{(1)}_{u,*})\cdot (f^{(1)}_b(T,\hat{B},\hat{\xi})-f^{(1)}_{b,*})
\cdot \mathcal{P}_N(V,T,\hat{U},\hat{B})\\ \notag
 &\times
\exp\Big\{-K_{m}^{(0)}(\tilde{v},\tilde{U})-L_{m}^{(0)}(\tilde{t},\tilde{B})+O(\log W/\sqrt{W})\Big\}
 \\ \notag
&\times \exp\Big\{f_u(V,\hat{U},\hat{\xi})+f_b(T,\hat{B},\hat{\xi})\Big\}\cdot \prod\limits_{j\in\Lambda}
\Big(a_+-a_-+O(\log W/\sqrt{W}))^2\\ \notag
&\times\prod\limits_{j\in\Lambda}
\Big(a_+-a_-+O(\log W/\sqrt{W}))^2.
\end{align}
Recall that at the $\delta$-neighborhood of the point $1$ we have
\begin{align*}
|f^{(1)}_b(T,\hat{B},\hat{\xi})-f^{(1)}_{b,*}|\le C\log W/\sqrt{W},\\
|f^{(1)}_u(V,\hat{U},\hat{\xi})-f^{(1)}_{u,*}|\le C\log W/\sqrt{W},
\end{align*}
which together with Lemma \ref{l:p_bound} gives $C\log^4 W/W^2$. Therefore,
to get non-zero contribution to (\ref{okr_1}) we can take only the first order of the
expansions of $f^{(1)}_b-f^{(1)}_{b,*}$,
$f^{(1)}_u-f^{(1)}_{u,*}$, the second order in the expansion of $\mathcal{P}_N$, and
zero orders in all other terms. Thus,
\begin{align}\label{okr_2}
&\big\langle (f^{(1)}_u-f^{(1)}_{u,*}) (f^{(1)}_b-f^{(1)}_{b,*})
\cdot \mathcal{P}_N\big\rangle_1=
\dfrac{W^2}{(8\pi^2)^{|\Lambda|}}\int\limits_{-\log W}^{\log W} d\tilde{\vphantom{b}u}d\tilde{b}\\ \notag
&\times\displaystyle\int_{0}^{2\pi} \prod\limits_{j\in\Lambda\setminus\{\overline{1}\}}
\dfrac{d\theta_{j}d\sigma_j}{(2\pi)^2}
\int\limits_{0}^{\log W}\prod\limits_{j\in\Lambda\setminus\{\overline{1}\}}(2\tilde{v}_j d\tilde{v}_j)(2\tilde{t}_j
d\tilde{t}_j)
\cdot (a_+-a_-)^{4|\Lambda|}\\
\notag
&\times (f^{(1)}_u(V,\hat{U},\hat{\xi})-f^{(1)}_{u,*})\cdot (f^{(1)}_b(T,\hat{B},\hat{\xi})-f^{(1)}_{b,*})
\cdot \mathcal{P}_N(V,T,\hat{U},\hat{B})\\ \notag
 &\times
\exp\Big\{-K_{m}^{(0)}(\tilde{v},\tilde{U})-L_{m}^{(0)}(\tilde{t},\tilde{B})\Big\}
 \cdot \exp\Big\{f_u(I,L_{\pm},\hat{\xi})+f_b(I,L_{\pm},\hat{\xi})\Big\}+o(1),
\end{align}
where $K_m^{(0)}(\tilde{v},\tilde{U})$ is defined in (\ref{K_0}).

We are left to compute the first order of expansions of $f^{(1)}_b-f^{(1)}_{b,*}$,
$f^{(1)}_u-f^{(1)}_{u,*}$ and the second order of the expansion of $\mathcal{P}_N$.
\begin{lemma}\label{l:comp}
We can write in the $\delta$-neighborhood of the point $1$:
\begin{align*}
%+c_j^{(v)}\dfrac{\tilde{v}_{j}}{\sqrt{W}}
&f^{(1)}_u(V,\hat{U},\hat{\xi})-f^{(1)}_{u,*}=\sum\limits_{j\in\Lambda}
\left(c_{1}\dfrac{\tilde{u}_{j,1}}{\sqrt{W}}+
c_{2}\dfrac{\tilde{u}_{j,2}}{\sqrt{W}}\right)+O(\log^2W/W),\\ \notag
&f^{(1)}_b(T,\hat{B},\hat{\xi})-f^{(1)}_{b,*}=\sum\limits_{j\in\Lambda}
d_{2}\dfrac{\tilde{b}_{j,2}}{\sqrt{W}}+O(\log^2W/W),\\
&\mdet \left(\alpha\Delta+I+\mathrm{diag}
\{U_j^{-1}\otimes B_j^{-1}\}_{j\in\Lambda}\right)
=\sum\limits_{j,k\in\Lambda}\left(a_{1,jk}
\dfrac{\tilde{u}_{j,1}\tilde{u}_{k,2}}{W}+a_{2,jk}
\dfrac{\tilde{u}_{j,1}\tilde{b}_{k,1}}{W}\right.\\
&\left.+a_{3,jk}
\dfrac{\tilde{u}_{j,2}\tilde{b}_{k,2}}{W}+a_{4,jk}\dfrac{\tilde{b}_{j,1}\tilde{b}_{k,2}}{W}+
a_{5,jk}
\dfrac{e^{i\theta_j}\tilde{v}_{j}\cdot e^{i\sigma_k}\tilde{t}_{k}}{W}\right)+
O(\log^3W/W^{3/2}),
\end{align*}
where
\begin{align}\label{cd}
c_{1}&=\dfrac{ia_+}{|\Lambda|\rho(\lambda_0)}\cdot \left(-\dfrac{i\,e^{ic_0\theta_\varepsilon}}{e^{ic_0\theta_\varepsilon}-e^{-ic_0\theta_\varepsilon}}
-\dfrac{2c_0\theta_\varepsilon}{(e^{ic_0\theta_\varepsilon}-e^{-ic_0\theta_\varepsilon})^2}\right),\\ \notag
c_{2}&=\dfrac{ia_-}{|\Lambda|\rho(\lambda_0)}\cdot\left(\dfrac{i\,e^{-ic_0\theta_\varepsilon}}{e^{ic_0\theta_\varepsilon}-e^{-ic_0\theta_\varepsilon}}
+\dfrac{2c_0\theta_\varepsilon}{(e^{ic_0\theta_\varepsilon}-e^{-ic_0\theta_\varepsilon})^2}\right),\\ \notag
 d_{2}&=\dfrac{ia_-}{|\Lambda|\rho(\lambda_0)},\quad a_{3,jk}=a_3:=i\cdot|\mdet
 \left(\alpha\Delta+(1+a_+^{-2})I\right)|^2
 \cdot \mdet (\alpha\Delta)_1^2.
\end{align}
Here $\mdet (\alpha\Delta)_1$ is a minor of $\alpha\Delta$ without the first row and column,
$\theta_\varepsilon$ and $c_0$ are defined in (\ref{theta}),
and $a_{1,jk}$, $a_{2,jk}$, $a_{4,jk}$, $a_{5,jk}$ are some constants.
\end{lemma}
\begin{remark}
We do not compute the values of $a_{1,jk}$, $a_{2,jk}$, $a_{4,jk}$, $a_{5,jk}$ since
we do not need them for further computations.
\end{remark}
\begin{proof} We are going to expand $f^{(1)}_u(V,\hat{U},\hat{\xi})-f^{(1)}_{u,*}$ in $\tilde{u}_{j,1}$,
$\tilde{u}_{j,2}$ and $\tilde{v}_j$ in the $\delta$-neighborhood of the point $1$ up to the first order.
Taking into account (\ref{f*}), the zero order is $0$. Compute now the first derivatives with respect to
 $\tilde{u}_{j,1}$, $\tilde{u}_{j,2}$ and $\tilde{v}_j$.

Write
\begin{align}\label{der1}
&\dfrac{\partial}{\partial \xi_1^\prime}\int d\mu(P_1)
\exp\Big\{-\dfrac{1}{|\Lambda|}\sum\limits_{j\in\Lambda}\Tr (V_jP_1^*)^*\hat{U}_j(V_jP_1^*)
(i\hat{\xi}_1/\rho(\lambda_0)-\varepsilon L)\Big\}\\ \notag
&=-\int \dfrac{i\,d\mu(P_1)}{\rho(\lambda_0)|\Lambda|}\Big(\sum\limits_{j\in\Lambda}u_{j,1}
-\sum\limits_{j\in\Lambda}|(V_jP_1^*)_{12}|^2(u_{j,1}-u_{j,2})\Big)\\ \notag
&\times \exp\Big\{\dfrac{1}{|\Lambda|}\sum\limits_{j\in\Lambda}\Tr (V_jP_1^*)^*\hat{U}_j(V_jP_1^*)
(i\hat{\xi}_1/\rho(\lambda_0)-\varepsilon L)\Big\}
\end{align}
Note that
\begin{multline*}
|(V_jP_1^*)_{12}|^2=v_j^2|(P_1)_{11}|^2+(1-v_j^2)|(P_1)_{12}|^2\\+(V_j)_{12}(\bar{P}_1)_{22}
(V_j)_{11}(\bar{P}_1)_{12}+(\bar{V}_j)_{12}(P_1)_{22}
(\bar{V}_j)_{11}(P_1)_{12}.
\end{multline*}
Since the integral (\ref{der1}) of the last two summands is zero,  (\ref{der1}) can be rewritten as
\begin{align*}
&\dfrac{\partial}{\partial \xi_1^\prime}\int d\mu(P_1)
\exp\Big\{-\dfrac{1}{|\Lambda|}\sum\limits_{j\in\Lambda}\Tr (V_jP_1^*)^*\hat{U}_j(V_jP_1^*)
(i\hat{\xi}_1/\rho(\lambda_0)-\varepsilon L)\Big\}\\ \notag
&=-\int \dfrac{i\,d\mu(P_1)}{\rho(\lambda_0)|\Lambda|}\sum\limits_{j\in\Lambda}\Big(u_{j,1}
-\left(v_j^2|(P_1)_{11}|^2+(1-v_j^2)|(P_1)_{12}|^2\right)(u_{j,1}-u_{j,2})\Big)\\ \notag
&\times \exp\Big\{-\dfrac{1}{|\Lambda|}\sum\limits_{j\in\Lambda}\Tr (V_jP_1^*)^*\hat{U}_j(V_jP_1^*)
(i\hat{\xi}_1/\rho(\lambda_0)-\varepsilon L)\Big\}.
\end{align*}
This expression depends only on $v_j^2$ and thus the first derivative of $f^{(1)}_u(V,\hat{U},\hat{\xi})$
with respect to $\tilde{v}_j$ is zero. Hence, we are left to compute
\[
\dfrac{\partial}{\partial \tilde{u}_{j,1}}f^{(1)}_u(I,\hat{U},\hat{\xi}),\quad
\dfrac{\partial}{\partial \tilde{u}_{j,2}}f^{(1)}_u(I,\hat{U},\hat{\xi}).
\]
According to (\ref{f_sm}) and Proposition \ref{p:Its-Zub} (i), we get
\begin{align}\label{f_u}
\exp\{f_u(I,\hat{U},\hat{\xi})\}&=\dfrac{e_1-e_2}{2i\theta_\varepsilon \sum\limits_{j\in \Lambda}
(u_{j,1}-u_{j,2})/|\Lambda|},
\end{align}
where
\begin{align*}
e_1&=\exp\Big\{-\frac{1}{|\Lambda|}\sum\limits_{j\in \Lambda}
(u_{j,1}(i\xi_1/\rho(\lambda_0)-\varepsilon)+u_{j,2}
(i\xi_2/\rho(\lambda_0)+\varepsilon))\Big\},\\
e_2&=\exp\Big\{-\frac{1}{|\Lambda|}\sum\limits_{j\in \Lambda}
(u_{j,2}(i\xi_1/\rho(\lambda_0)-\varepsilon)+u_{j,1}
(i\xi_2/\rho(\lambda_0)+\varepsilon))\Big\}.\\
\end{align*}
Hence,
\begin{equation*}
f^{(1)}_u(I,\hat{U},\hat{\xi})=-\dfrac{i}{|\Lambda|\rho(\lambda_0)}\cdot\dfrac{\sum_{j\in \Lambda}
(u_{j,1} e_1 -u_{j,2} e_2)}{e_1-e_2}+\dfrac{1}{2\theta_\varepsilon\rho(\lambda_0)}.
\end{equation*}
Taking the derivatives in $\tilde{u}_{j,1}$, $\tilde{u}_{j,2}$ we get the expressions (\ref{cd}) for $c_1$, $c_2$.

Expand now $f^{(1)}_b(T,\hat{B},\hat{\xi})-f^{(1)}_{b,*}$ with respect to $\tilde{b}_{j,1}$,
$\tilde{b}_{j,2}$ and $\tilde{t}_j$ in the $\delta$-neighborhood of the point $1$ up to the first order.
By the same argument as above we get that zero order is $0$, and the first derivative of $f^{(1)}_b(T,\hat{B},\hat{\xi})$
with respect to $\tilde{t}_j$ is zero. Thus, we need to compute
\[
\dfrac{\partial}{\partial \tilde{b}_{j,1}}f^{(1)}_b(I,\hat{B},\hat{\xi}),\quad
\dfrac{\partial}{\partial \tilde{b}_{j,2}}f^{(1)}_b(I,\hat{B},\hat{\xi}).
\]
Taking into account Proposition \ref{p:Its-Zub} (ii), we have
\begin{multline*}
\exp\{f_b(I,\hat{B},\hat{\xi})\}\\=\dfrac{\exp\Big\{\dfrac{1}{|\Lambda|}
\sum\limits_{j\in \Lambda}(b_{j,1}(i\xi_1/\rho(\lambda_0)-\varepsilon)-b_{j,2}
(i\xi_2/\rho(\lambda_0)+\varepsilon))\Big\}}{2i\theta_\varepsilon \sum\limits_{j\in \Lambda}
(b_{j,1}+b_{j,2})/|\Lambda|},
\end{multline*}
and hence
\begin{equation}\label{f_b1}
f_b^{(1)}(I,\hat{B},\hat{\xi})=-\dfrac{i}{|\Lambda|\rho(\lambda_0)}\sum\limits_{j\in \Lambda}b_{j,2}-
\dfrac{1}{2\theta_\varepsilon\rho(\lambda_0)}.
\end{equation}
Taking the derivatives with respect to $\tilde{b}_{j,1}$, $\tilde{b}_{j,2}$, we get the assertion
of the lemma for $f^{(1)}_b(T,\hat{B},\hat{\xi})-f^{(1)}_{b,*}$.
%Hence the first order in the expansion of $f^{(1)}_u(V,\hat{U},\hat{\xi})$ and
%$f^{(1)}_b(T,\hat{B},\hat{\xi})$ around the saddle has the form
%\begin{align}\label{f_order}
%\sum\limits_{j\in\Lambda}\left(c_{j,1}\dfrac{\tilde{u}_{j,1}}{\sqrt{W}}+
%c_{j,2}\dfrac{\tilde{u}_{j,2}}{\sqrt{W}}\right),\\ \notag
%\sum\limits_{j\in\Lambda}\left(d_{j,1}\dfrac{\tilde{b}_{j,1}}{\sqrt{W}}+
%d_{j,2}\dfrac{\tilde{b}_{j,2}}{\sqrt{W}}\right),
%\end{align}

To complete the proof of the lemma we need to expand $$\mdet \left(\alpha\Delta+I+\hbox{diag}
\{U_j^{-1}\otimes B_j^{-1}\}_{j\in\Lambda}\right)$$ near the saddle-point 1 up to the second order.
According to Lemma \ref{l:p_bound} the first and zero order are $0$. Let us show now that all the
second partial derivatives of that determinant are zero except
\begin{equation}\label{non-zero}
\dfrac{\partial^2}{\partial \tilde{u}_{j,1}\partial \tilde{u}_{k,2}},\quad
\dfrac{\partial^2}{\partial \tilde{u}_{j,1}\partial \tilde{b}_{k,1}},\quad
\dfrac{\partial^2}{\partial \tilde{u}_{j,2}\partial \tilde{b}_{k,2}},\quad
\dfrac{\partial^2}{\partial \tilde{b}_{j,1}\partial \tilde{b}_{k,2}},
\quad \dfrac{\partial^2}{\partial \tilde{v}_{j}\partial \tilde{t}_{k}}.
\end{equation}
Indeed, let us first take only the second partial derivatives with respect to $\tilde{u}_{1}$, $\tilde{u}_{2}$,
$\tilde{b}_{1}$ or $\tilde{b}_{2}$. Then we can put $V_j=T_j=I$. We get
\begin{align*}
&\mdet \left(\alpha\Delta+I+\hbox{diag}
\{U_j^{-1}\otimes B_j^{-1}\}_{j\in\Lambda}\right)\Big|_{V_j=T_j=I}\\
&=
\mdet \left(\alpha\Delta+I+\hbox{diag}
\{u_{j,1}^{-1} b_{j,1}^{-1}\}_{j\in\Lambda}\right) \mdet \left(\alpha\Delta+I-\hbox{diag}
\{u_{j,2}^{-1} b_{j,2}^{-1}\}_{j\in\Lambda}\right)\\
&\times
\mdet \left(\alpha\Delta+I-\hbox{diag}
\{u_{j,1}^{-1} b_{j,2}^{-1}\}_{j\in\Lambda}\right)
\mdet \left(\alpha\Delta+I+\hbox{diag}
\{u_{j,2}^{-1} b_{j,1}^{-1}\}_{j\in\Lambda}\right)
\end{align*}
But the last two determinant are zero at the saddle-point $1$. Thus, the non-zero second derivative
can be obtain only if we differentiate once each of that brackets. Hence, from such derivatives
only (\ref{non-zero}) are non-zero.

Consider now the second partial derivatives of the determinant which contain derivatives over $\tilde{v}$
or $\tilde{t}$, but not both of them. Without loss of generality, let it contains the derivative over
$\tilde{v}$ only. Then we can put $T=I$ and write
\begin{align*}
&\mdet \left(\alpha\Delta+I+\hbox{diag}
\{U_j^{-1}\otimes B_j^{-1}\}_{j\in\Lambda}\right)\Big|_{T_j=I}\\
&=\mdet \left(\alpha\Delta+I+\hbox{diag}
\{b_{j,1}^{-1}U_j^{-1}\}_{j\in\Lambda}\right)
 \mdet \left(\alpha\Delta+I-\hbox{diag}
\{b_{j,2}^{-1}U_j^{-1}\}_{j\in\Lambda}\right).
\end{align*}
Since both determinants are zero at the saddle-point $1$, to get a non-zero second partial derivative
we must differentiate each of them once. But it is easy to see that the first partial derivative
of both determinants with respect to $\tilde{v}$ is zero at the saddle-point 1.
Hence, all the second partial derivatives of the determinant except (\ref{non-zero})  are zero.

Compute now
\begin{equation*}
a_{3,jk}/W=\dfrac{\partial^2}{\partial\tilde{u}_{j,2}\partial \tilde{b}_{k,2}}
\mdet \left(\alpha\Delta+I+\hbox{diag}
\{U_j^{-1}\otimes B_j^{-1}\}_{j\in\Lambda}\right).
\end{equation*}
for each $j,k\in \Lambda$.
We have
\begin{align*}\notag
&a_{3,jk}/W=\dfrac{\partial^2}{\partial\tilde{u}_{j,2}\partial \tilde{b}_{k,2}}
\mdet \left(\alpha\Delta+I+\hbox{diag}
\{U_j^{-1}\otimes B_j^{-1}\}_{j\in\Lambda}\right)\Big|_{V=T=I, \hat{U}=\hat{B}=L_\pm}\\
&=\dfrac{\partial^2}{\partial\tilde{u}_{j,2}\partial \tilde{b}_{k,2}}
\mdet \left(\alpha\Delta+I+\hbox{diag}
\{\hat{U}_j^{-1}\otimes \hat{B}_j^{-1}\}_{j\in\Lambda}\right)\Big|_{\hat{U}=\hat{B}=L_\pm}\\ \notag
&=\dfrac{\partial^2}{\partial\tilde{u}_{j,2}\partial \tilde{b}_{k,2}}
\prod\limits_{s_1,s_2=1}^2\mdet \left(\alpha\Delta+I+\hbox{diag}
\{u_{j,s_1}^{-1}\otimes (-1)^{s_2+1}b_{j,s_2}^{-1}\}_{j\in\Lambda}\right)\Big|_{\hat{U}=\hat{B}=L_\pm}\\
\notag
&=i\cdot|\mdet \left(\alpha\Delta+(1+a_+^{-2})I\right)|^2\cdot \mdet (\alpha\Delta)_1^2/W,
\end{align*}
which completes the proof of the lemma.$\quad\Box$
\end{proof}
Besides, substituting $\hat{U}_j=\hat{B}_j=L_\pm$ to (\ref{f_u}) -- (\ref{f_b1}), we get
\begin{align}\label{at_saddle}
&\exp\{f_u(I,L_{\pm},\hat{\xi})\}=e^{\lambda_0(\xi_1+\xi_2)/2\rho(\lambda_0)}\cdot
\dfrac{e^{ic_0\theta_\varepsilon}-e^{-ic_0\theta_\varepsilon}}
{2c_0i\theta_\varepsilon},\\ \notag
&\exp\{f_b(I,L_{\pm},\hat{\xi})\}=e^{-\lambda_0(\xi_1+\xi_2)/2\rho(\lambda_0)}\cdot
\dfrac{e^{-ic_0\theta_\varepsilon}}
{2c_0i\theta_\varepsilon},\\ \notag
&f^{(1)}_{u,*}=\dfrac{1}{2\rho(\lambda_0)\theta_\varepsilon}-
\dfrac{i}{\rho(\lambda_0)} \cdot\dfrac{a_+\,e^{ic_0\theta_\varepsilon}-a_-\,
e^{-ic_0\theta_\varepsilon}}{e^{ic_0\theta_\varepsilon}-e^{-ic_0\theta_\varepsilon}},\\ \notag
&f^{(1)}_{b,*}=\dfrac{ia_-}{\rho(\lambda_0)}-\dfrac{1}{2\theta_\varepsilon\rho(\lambda_0)},
\end{align}
where $\theta_\varepsilon$ and $c_0$ are defined in (\ref{theta}).

Since the Gaussian integral of the linear term is zero, substituting (\ref{at_saddle}) into
(\ref{okr_2}) and using Lemma \ref{l:comp}, we obtain
\begin{equation*}
W^{-1}\langle (f^{(1)}_u(V,\hat{U},\hat{\xi})-f^{(1)}_{u,*})\cdot (f^{(1)}_b(T,\hat{B},\hat{\xi})-f^{(1)}_{b,*})
\cdot \tilde{\mathcal{P}}_N(V,T,\hat{U},\hat{B})\rangle_1=o(1)
\end{equation*}
with $\tilde{\mathcal{P}}_N$ of (\ref{P_n_1}) -- (\ref{p_cal_til}), and thus
\begin{align}\notag
&\big\langle (f^{(1)}_u-f^{(1)}_{u,*}) (f^{(1)}_b-f^{(1)}_{b,*})
\cdot \mathcal{P}_N\big\rangle_1=
\dfrac{1}{(8\pi^2)^{|\Lambda|}}\int\limits_{-\log W}^{\log W} \prod\limits_{j\in\Lambda}
 d\tilde{u}_{j,1}d\tilde{u}_{j,2}
d \tilde{b}_{j,1}d\tilde{b}_{j,2}\\ \label{okr_3}
&\times\displaystyle\int_{0}^{2\pi} \prod\limits_{j\in\Lambda\setminus\{\overline{1}\}}
\dfrac{d\theta_{j}d\sigma_j}{(2\pi)^2}
\int\limits_{0}^{\log W}\prod\limits_{j\in\Lambda\setminus\{\overline{1}\}}(2\tilde{v}_j d\tilde{v}_j)(2\tilde{t}_j
d\tilde{t}_j)\\ \notag
&\times  (2\pi\rho(\lambda_0))^{4|\Lambda|}\cdot \dfrac{e^{ic_0\theta_\varepsilon}-e^{-ic_0\theta_\varepsilon}}
{2ic_0\theta_\varepsilon}\cdot
\dfrac{e^{-ic_0\theta_\varepsilon}}
{2ic_0\theta_\varepsilon}\cdot c_2d_2a_3,
\\ \notag
&\times \sum\limits_{j_1,k_1\in\Lambda}\sum\limits_{j_2,k_2\in\Lambda} \tilde{u}_{j_1,2}
\tilde{u}_{k_1,2}\tilde{b}_{j_2,2}
\tilde{b}_{k_2,2}\cdot
\exp\Big\{-K_{m}^{(0)}(\tilde{v},\tilde{U})-L_{m}^{(0)}(\tilde{t},\tilde{B})\Big\}+o(1).
\end{align}
Moreover,
\begin{align*}
&\int_0^{\log W} \prod\limits_{j\in\Lambda\setminus\{\overline{1}\}}(2\tilde{v}_jd\tilde{v}_jd\theta_j/2\pi)\exp\{-\alpha (a_+-a_-)^2\sum\limits_{j\sim j^\prime}
|\tilde{v}_je^{i\theta_j}-\tilde{v}_{j^\prime}e^{i\theta_{j^\prime}}|^2\}\\
&=\pi^{-|\Lambda|+1}\int\limits_{-\infty}^{\infty}\prod\limits_{j\in\Lambda\setminus\{\overline{1}\}}
d\phi_{j,1}d\phi_{j,2}\
\exp\{-\alpha (a_+-a_-)^2(\nabla \phi )^2\}+o(1)\\
&=(2\pi\rho(\lambda_0))^{-2|\Lambda|+2}\mdet^{-1}
(\alpha\Delta)_1+o(1),
\end{align*}
where $\phi_1=(0,0)$, $\phi_j=(\phi_{j,1},\phi_{j,2})$ for $j\in\Lambda\setminus\{\overline{1}\}$.
The first equality here is obtained by changing
\[
\phi_{j,1}=\tilde{v}_j\cos \theta_j,\quad\phi_{j,2}=\tilde{v}_j\sin \theta_j.
\]
The same expression can be obtained for the integral over $\tilde{t_j}$.

Substituting this and (\ref{cd}) to (\ref{okr_3}), we get
\begin{align*}
&2^{|\Lambda|}\langle (f^{(1)}_u(V,\hat{U},\hat{\xi})-f^{(1)}_{u,*}) (f^{(1)}_b(T,\hat{B},\hat{\xi})-f^{(1)}_{b,*})
\cdot \mathcal{P}_N\rangle_1=\dfrac{i}{(2\pi)^{2|\Lambda|}}
\int\limits_{-\infty}^{\infty}  d\tilde{u}d\tilde{b}\\
 &\times (-a_-^2)\left(ie^{-2ic_0\theta_\varepsilon}+
\dfrac{2c_0\theta_\varepsilon e^{-ic_0\theta_\varepsilon}}
{e^{ic_0\theta_\varepsilon}-e^{-ic_0\theta_\varepsilon}}\right)
\cdot (2ic_0\theta_\varepsilon\rho(\lambda_0))^{-2}\\
\notag
&\times \dfrac{1}{|\Lambda|^2}\sum\limits_{j_1,k_1\in\Lambda}\sum\limits_{j_2,k_2\in\Lambda}
|\mdet \left(\alpha\Delta+(1+a_+^{-2})I\right)|^2 \cdot\tilde{u}_{j_1,2}
\tilde{u}_{k_1,2}\tilde{b}_{j_2,2}\tilde{b}_{k_2,2}\\ \notag
 &\times
\exp\Big\{-a_+^2 (M_+\tilde{u}_1,\tilde{u}_1)/2 -a_+^2 (M_+\tilde{b}_1,\tilde{b}_1)/2\Big\}\cdot (2\pi\rho(\lambda_0))^{4}\\ \notag
&\times\exp\Big\{-a_-^2 (M_-\tilde{u}_2,\tilde{u}_2)/2-a_-^2 (M_-\tilde{b}_2,\tilde{b}_2)/2\Big\}+o(1),
\end{align*}
where $M_\pm=\alpha\Delta+(1+a_\pm^{-2})I$. Thus, taking the Gaussian integral, we obtain
\begin{align*}
&2^{|\Lambda|}\langle (f^{(1)}_u(V,\hat{U},\hat{\xi})-f^{(1)}_{u,*})
(f^{(1)}_b(T,\hat{B},\hat{\xi})-f^{(1)}_{b,*})
\cdot \mathcal{P}_N\rangle_1=\dfrac{a_-^2c_0^2}{(2\theta_\varepsilon\rho(\lambda_0))^{2}}\\
&\times\Big(\dfrac{2ic_0\theta_\varepsilon e^{-ic_0\theta_\varepsilon}}
{e^{ic_0\theta_\varepsilon}-e^{-ic_0\theta_\varepsilon}}-e^{-2ic_0\theta_\varepsilon}\Big)\cdot
 \dfrac{1}{|\Lambda|^2}\sum\limits_{j_1,k_1\in\Lambda}(M_-)^{-1}_{j_1k_1}
\sum\limits_{j_2,k_2\in\Lambda}(M_-)^{-1}_{j_2k_2}+o(1).
\end{align*}
This and
\[
\sum\limits_{j\in\Lambda}(M_-)^{-1}_{jk}=
\sum\limits_{j\in\Lambda}(\alpha\Delta+(1+a_-^{-2})I)^{-1}_{jk}=(1+a_-^{-2})^{-1}=(-c_0a_-)^{-1}
\]
give finally Lemma \ref{l:main_contr}.
%\begin{multline*}
%2^{|\Lambda|}\langle (f^{(1)}_u(V,\hat{U},\hat{\xi})-f^{(1)}_{u,*})
% (f^{(1)}_b(T,\hat{B},\hat{\xi})-f^{(1)}_{b,*})
%\cdot \mathcal{P}_N\rangle_1\\
%=\dfrac{1}{(2\theta_\varepsilon\rho(\lambda_0))^2}\cdot\left(e^{-ic_0\theta_{\varepsilon}}-
%\dfrac{2ic_0\theta_\varepsilon e^{-ic_0\theta_{\varepsilon}}}
%{e^{ic_0\theta_{\varepsilon}}-e^{-ic_0\theta_{\varepsilon}}}\right),
%\end{multline*}
\end{proof}
In addition, substituting (\ref{at_saddle}), we obtain
\begin{align*}
&f^{(1)}_{u,*}ia_-/\rho(\lambda_0)
-f^{(1)}_{b,*}ia_+/\rho(\lambda_0)-f^{(1)}_{b,*}f^{(1)}_{u,*}\\ \notag
&=-\dfrac{1}{\rho^2(\lambda_0)}-\dfrac{ic_0e^{-ic_0\theta_\varepsilon}}
{2\theta_\varepsilon\rho(\lambda_0)^2(e^{ic_0\theta_\varepsilon}-e^{-ic_0\theta_\varepsilon})}+
\dfrac{1}{(2\rho(\lambda_0)\theta_\varepsilon)^2}.
\end{align*}
Combining this with Lemma \ref{l:main_contr} we get
\begin{align*}
\dfrac{\partial^2}{\partial\xi_1^\prime\partial\xi_2^\prime}G_2^{+-}(z,\xi)\Big|_{\xi^\prime=\xi}=
-\dfrac{1}{\rho^2(\lambda_0)}
+\dfrac{1-e^{-2ic_0\theta_\varepsilon}}{(2\rho(\lambda_0)\theta_\varepsilon)^2}
+o(1),
\end{align*}
thus
\begin{multline*}
\dfrac{\partial^2}{\partial\xi_1^\prime\partial\xi_2^\prime}\Big(G_2^{+-}(z,\xi)+
\overline{G}_2^{\,+-}(z,\xi)\Big)\Big|_{\xi^\prime=\xi}\\=
-\dfrac{2}{\rho^2(\lambda_0)}
-\dfrac{\big(e^{ic_0\theta_\varepsilon}-e^{-ic_0\theta_\varepsilon}\big)^2}{(2\rho(\lambda_0)\theta_\varepsilon)^2}
+o(1),
\end{multline*}
which gives (\ref{+-}).

\subsection{Proof of Theorem \ref{thm:++}}
Now let us compute $G_2^{++}(z_1,z_2)$. Again Lemmas \ref{l:kont_u} -- \ref{l:kont_b_+} yield,
that the main contribution is given by the $\delta$-neighborhoods of the points
\begin{enumerate}
    \item $\hat{A}_j=L_{+}$, $\hat{U}_j=L_{\pm}\,$ or $\,\hat{U}_j=L_{\mp}$, $j\in\Lambda$, and
    $$
    |(V_j)_{12}|=\left\{\begin{array}{ll}
    0,& \hat{U}_j=\hat{U}_1,\\
    1,& \mathrm{otherwise}.
    \end{array}\right.
    $$

    \item $\hat{A}_j=L_{+}$, $\hat{U}_j=L_{+}$, $j\in\Lambda$.

    \item $\hat{A}_j=L_{+}$, $\hat{U}_j=L_{-}$, $j\in\Lambda$.
\end{enumerate}
Using the same idea as for $G_2^{+-}(z_1,z_2)$, we get
\begin{align*}
\langle f_u^{(1)}f_a^{(1)}\cdot \mathcal{P}_N\rangle^{(+)}&=
\langle f_u^{(1)}(f_a^{(1)}-f_a^*)\cdot \mathcal{P}_N\rangle^{(+)}+f_a^*
\langle f_u^{(1)}\cdot \mathcal{P}_N\rangle^{+}\\
&=
\langle f_u^{(1)}(f_a^{(1)}-f_a^*)\cdot \mathcal{P}_N\rangle^{(+)}+f_a^*
(-ia_-+o(1)),
\end{align*}
where $\langle\ldots\rangle^{(+)}$ is $\langle\ldots\rangle$ for $G_2^{++}(z_1,z_2)$, and
\begin{align*}\notag
\exp\{f_a(\tilde{V},\hat{A},\hat{\xi})\}&:=\int
e^{-|\Lambda|^{-1}\sum_{j\in\Lambda}\Tr (\tilde{V}_j\tilde{P}_1^*)^*\hat{A}_j
(\tilde{V}_j\tilde{P}_1^*)
(i\hat{\xi}/\rho(\lambda_0)-\varepsilon)}d\mu(\tilde{P}_1),\\
f_a^{(1)}(\tilde{V},\hat{A},\hat{\xi})&:=\dfrac{\partial}{\partial \xi_2^\prime}
f_a(\tilde{V},\hat{A},\hat{\xi}_1)\Big|_{\xi^\prime=\xi}.
\end{align*}
Note that for all saddle-points 1 -- 3
\[
f_a^{(1)}(\tilde{V},\hat{A},\hat{\xi})\Big|_s=f_a^*=-\dfrac{ia_+}{\rho(\lambda_0)}
\]
Repeating almost literally the proof of Lemma \ref{l:trick}, we get
\[
|\langle f_u^{(1)}(f_a^{(1)}-f_a^*)\cdot \mathcal{P}_N\rangle^{(+)}_{1,2,3}|\le \log W/\sqrt{W},
\]
and hence
\begin{align*}
\dfrac{\partial^2}{\partial\xi_1^\prime\partial\xi_2^\prime} G_2^{++}(z,\xi)
\Big|_{\xi^\prime=\xi}&=\langle f_u^{(1)}f_a^{(1)}\cdot \mathcal{P}_N\rangle^{(+)}+o(1)=f_a^*\cdot
\langle f_u^{(1)}\cdot \mathcal{P}_N\rangle^{(+)}+o(1)\\
&
=ia_+f_a^*+o(1)=a_+^2/\rho(\lambda)^2+o(1),
\end{align*}
which yields (\ref{++}).

\section{Appendix}
\subsection{Grassmann integration}
Let us consider two sets of formal variables
$\{\psi_j\}_{j=1}^n,\{\overline{\psi}_j\}_{j=1}^n$, which satisfy the anticommutation
conditions
\begin{equation}\label{anticom}
\psi_j\psi_k+\psi_k\psi_j=\overline{\psi}_j\psi_k+\psi_k\overline{\psi}_j=\overline{\psi}_j\overline{\psi}_k+
\overline{\psi}_k\overline{\psi}_j=0,\quad j,k=1,\ldots,n.
\end{equation}
Note that this definition implies $\psi_j^2=\overline{\psi}_j^2=0$.
These two sets of variables $\{\psi_j\}_{j=1}^n$ and $\{\overline{\psi}_j\}_{j=1}^n$ generate the Grassmann
algebra $\mathfrak{A}$. Taking into account that $\psi_j^2=0$, we have that all elements of $\mathfrak{A}$
are polynomials of $\{\psi_j\}_{j=1}^n$ and $\{\overline{\psi}_j\}_{j=1}^n$ of degree at most one
in each variable. We can also define functions of
the Grassmann variables. Let $\chi$ be an element of $\mathfrak{A}$, i.e.
\begin{equation}\label{chi}
\chi=a+\sum\limits_{j=1}^n (a_j\psi_j+ b_j\overline{\psi}_j)+\sum\limits_{j\ne k}
(a_{j,k}\psi_j\psi_k+
b_{j,k}\psi_j\overline{\psi}_k+
c_{j,k}\overline{\psi}_j\overline{\psi}_k)+\ldots.
\end{equation}
For any
sufficiently smooth function $f$ we define by $f(\chi)$ the element of $\mathfrak{A}$ obtained by substituting $\chi-a$
in the Taylor series of $f$ at the point $a$. Since $\chi$ is a polynomial of $\{\psi_j\}_{j=1}^n$,
$\{\overline{\psi}_j\}_{j=1}^n$ of the form (\ref{chi}), according to (\ref{anticom}) there exists such
$l$ that $(\chi-a)^l=0$, and hence the series terminates after a finite number of terms and so $f(\chi)\in \mathfrak{A}$.

For example, we have
\begin{align}\notag
 &\exp\{a\,\overline{\psi}_1\psi_1\}=1+a\,\overline{\psi}_1\psi_1+(a\,\overline{\psi}_1\psi_1)^2/2+\ldots
 =1+a\,\overline{\psi}_1\psi_1,\\ \notag
&\exp\{a_{11}\overline{\psi}_1\psi_1+a_{12}\overline{\psi}_1\psi_2+
a_{21}\overline{\psi}_2\psi_1+a_{22}\overline{\psi}_2\psi_2\}=1+ a_{11}\overline{\psi}_1\psi_1\\
\label{ex_12} &+a_{12}\overline{\psi}_1\psi_2+ a_{21}\overline{\psi}_2\psi_1+a_{22}\overline{\psi}_2\psi_2+
(a_{11}\overline{\psi}_1\psi_1+a_{12}\overline{\psi}_1\psi_2\\ \notag &+
a_{21}\overline{\psi}_2\psi_1+a_{22}\overline{\psi}_2\psi_2)^2/2+\ldots=1+
a_{11}\overline{\psi}_1\psi_1+a_{12}\overline{\psi}_1\psi_2+ a_{21}\overline{\psi}_2\psi_1\\ \notag
&+a_{22}\overline{\psi}_2\psi_2+(a_{11}a_{22}-a_{12}a_{21})\overline{\psi}_1\psi_1\overline{\psi}_2\psi_2.
\end{align}
Following Berezin \cite{Ber}, we define the operation of
integration with respect to the anticommuting variables in a formal
way:
\begin{equation}\label{int_gr}
\intd d\,\psi_j=\intd d\,\overline{\psi}_j=0,\quad \intd
\psi_jd\,\psi_j=\intd \overline{\psi}_jd\,\overline{\psi}_j=1,
\end{equation}
and then extend the definition to the general element of $\mathfrak{A}$ by
the linearity. A multiple integral is defined to be a repeated
integral. Assume also that the ``differentials'' $d\,\psi_j$ and
$d\,\overline{\psi}_k$ anticommute with each other and with the
variables $\psi_j$ and $\overline{\psi}_k$. Thus, according to the definition, if
$$
f(\psi_1,\ldots,\psi_k)=p_0+\sum\limits_{j_1=1}^k
p_{j_1}\psi_{j_1}+\sum\limits_{j_1<j_2}p_{j_1,j_2}\psi_{j_1}\psi_{j_2}+
\ldots+p_{1,2,\ldots,k}\psi_1\ldots\psi_k,
$$
then
\begin{equation}\label{int}
\intd f(\psi_1,\ldots,\psi_k)d\,\psi_k\ldots d\,\psi_1=p_{1,2,\ldots,k}.
\end{equation}

   Let $A$ be an ordinary Hermitian matrix with positive real part. The following Gaussian
integral is well-known
\begin{equation}\label{G_C}
\intd \exp\Big\{-\sum\limits_{j,k=1}^nA_{jk}z_j\overline{z}_k\Big\} \prod\limits_{j=1}^n\dfrac{d\,\Re
z_jd\,\Im z_j}{\pi}=\dfrac{1}{\mdet A}.
\end{equation}
One of the important formulas of the Grassmann variables theory is the analog of this formula for the
Grassmann algebra (see \cite{Ber}):
\begin{equation}\label{G_Gr}
\int \exp\Big\{-\sum\limits_{j,k=1}^nA_{jk}\overline{\psi}_j\psi_k\Big\}
\prod\limits_{j=1}^nd\,\overline{\psi}_jd\,\psi_j=\mdet A,
\end{equation}
where $A$ now is any $n\times n$ matrix.

For $n=1$ and $2$ this formula  follows immediately from (\ref{ex_12}) and (\ref{int}).

Let
\[
F=\left(\begin{array}{cc}
a&\rho\\
\tau&b\\
\end{array}\right),\quad
\]
where $a$ and $b>0$ are Hermitian complex $k\times k$ matrices and $\rho$, $\tau$ are $k\times k$ matrices
of independent
anticommuting Grassmann variables, and let
\[
\Phi=(\psi_{1},\ldots,\psi_{k},z_{1},\ldots,z_{k})^t,
\]
where $\{\psi_j\}_{j=1}^k$ are independent Grassmann variables and $\{z_j\}_{j=1}^k$ are complex variables.
Combining (\ref{G_C}) -- (\ref{G_Gr}) we obtain (see \cite{Ber})
\begin{equation}\label{G_comb}
\intd \exp\{-\Phi^+F\Phi\}\prod\limits_{j=1}^kd\overline{\psi}_j\, d\psi_j \prod\limits_{j=1}^k\dfrac{\Re z_j\Im
z_j}{\pi}=\hbox{sdet}\, F,
\end{equation}
where
\begin{equation}\label{sdet}
 \hbox{sdet}\, F=\dfrac{\det\,(a-\rho\, b^{-1}\,\tau)}{\det\, b}.
\end{equation}
We will also need
\begin{proposition}({\bf see \cite{SupB:08} and references therein})\label{p:supboz}\\
Let $F$ be some function that depends only on combinations
\begin{align*}
\bar{\psi}\psi&:=\Big\{\sum\limits_{\alpha=1}^n \bar{\psi}_{j\alpha}\psi_{k\alpha}\Big\}_{j,k=1}^p,\quad
\bar{\psi}\phi:=\Big\{\sum\limits_{\alpha=1}^n \bar{\psi}_{j\alpha}\phi_{k\alpha}\Big\}_{j,k=1}^p,\\
\bar{\phi}\psi&:=\Big\{\sum\limits_{\alpha=1}^n \bar{\phi}_{j\alpha}\psi_{k\alpha}\Big\}_{j,k=1}^p,\quad
\bar{\phi}\phi:=\Big\{\sum\limits_{\alpha=1}^n \bar{\phi}_{j\alpha}\phi_{k\alpha}\Big\}_{j,k=1}^p,
\end{align*}
and set
\[
d\Psi=\prod\limits_{j=1}^p\prod\limits_{\alpha=1}^nd\bar{\psi}_{j\alpha} d\psi_{j\alpha},\quad
d\Phi=\prod\limits_{j=1}^p\prod\limits_{\alpha=1}^n\pi^{-1} d\Re \phi_{j\alpha} d\Im \phi_{j\alpha}.
\]
Assume also that $n\ge p$. Then
\begin{equation*}
\int F\left(\begin{array}{cc}
\bar{\psi}\psi & \bar{\psi}\phi\\
\bar{\phi}\psi & \bar{\phi}\phi
\end{array}
\right)d\Phi d\Psi=(i\pi)^{-p(p-1)}\int F(Q)\cdot \hbox{sdet\,}^n Q \,dQ,
\end{equation*}
where
\[
Q=\left(\begin{array}{ll}
U&\rho\\
\tau& B
\end{array}
\right),
\]
$U$ is $p\times p$ unitary matrix, $B$ is $p\times p$ positive Hermitian matrix, and
$\rho$, $\tau$ are $p\times p$ matrices whose entries are independent Grassmann variables.
Here
\[
dQ=d\rho d\tau dU dB,
\]
\begin{align*}
d\rho d\tau&=\prod\limits_{j,k=1}^pd\rho_{jk}\, d\tau_{jk},\\
dB&=\mathbf{1}_{B>0}\prod\limits_{j=1}^pdB_{jj}\prod\limits_{j,k=1}^pd\Re B_{jk} d\Im B_{jk},\\
dU&=\dfrac{\pi^{p(p-1)/2}}{\prod_{s=1}^p s!}\prod\limits_{j=1}^p\dfrac{du_j}{2\pi i}\cdot \Delta(u_1,\ldots,u_p)^2d\mu(V),
\end{align*}
where $u_j\in \mathbb{T}$ are the eigenvalues of $U$,
\[
\Delta(u_1,\ldots,u_p)=\prod\limits_{j<k} (u_j-u_k),
\]
 $V$ is a matrix diagonalizing $U$,
and $d\mu(V)$ is the normalized Haar measure over
$U(p)/U(1)^p$.
\end{proposition}
\subsection{Integration over $\mathring{U}(2)$ and $\mathring{U}(1,1)$}
The integral over $\mathring{U}(2)$ and $\mathring{U}(1,1)$ can be computed using
\begin{proposition}\label{p:Its-Zub}
\begin{enumerate}
\item[(i)] The normalized Haar measure $d\mu(U)$ over $\mathring{U}(2)$ can be parameterize as
follows
\begin{align}\label{mu}
&U=\left(\begin{array}{ll}
w& v\,e^{i\theta}\\
-v\,e^{-i\theta}& w\\
\end{array}\right),\quad
w=(1-v^2)^{1/2}\\ \notag
& d\mu(U)=\dfrac{d\theta}{2\pi}\cdot
(2v dv),\quad v\in [0,1],\quad \theta\in [0,2\pi].
\end{align}
If $C=\hbox{diag}\{c_1,c_2\}$ and
$D=\hbox{diag}\{d_1,d_2\}$, then for any $r$ we have
\begin{multline}\label{Its-Zub}
\int_{\mathring{U}(2)} \exp\{r\Tr CU^*DU\} d\mu(U)\\=
\dfrac{\exp\{r(c_1d_1+c_2t_2)\}-\exp\{r(c_1d_2+c_2t_1)\}}{r(c_1-c_2)(d_1-d_2)}.
\end{multline}

\item[(ii)] The Haar measure $d\nu(T)$ over the group $\mathring{U}(1,1)$ can be parameterize as
follows
\begin{align}\label{nu}
&T=\left(\begin{array}{ll}
s& t\,e^{i\sigma}\\
t\,e^{-i\sigma}& s\\
\end{array}\right),\quad
s=(1+t^2)^{1/2}\\ \notag
& d\nu(T)=\dfrac{d\sigma}{2\pi}\cdot
(2t dt),\quad t\in [0,\infty),\quad \sigma\in [0,2\pi].
\end{align}
If $C=\hbox{diag}\{c_1,c_2\}$ and
$D=\hbox{diag}\{d_1,d_2\}$, then for any $r$ such that $$\Re r(c_1-c_2)(d_1-d_2)>0$$
we have
\begin{equation*}
\int_{\mathring{U}(1,1)} \exp\{-r\Tr CT^{-1}DT\} d\,\nu(T)=
\dfrac{\exp\{-r(c_1d_1+c_2d_2)\}\}}{r(c_1-c_2)(d_1-d_2)}.
\end{equation*}
\end{enumerate}
\end{proposition}
Formula (\ref{Its-Zub}) is the well-known Harish Chandra/Itsykson-Zuber formula
(see e.g. \cite{Me:91}, Appendix 5). The proof of (ii) can be found e.g. in \cite{F:02},
Appendix C.

\bigskip

\textbf{Acknowledgement.}
Supported by NSF grant DMS 1128155. This research also was partially supported by
RF Government grant 11.G34.31.0026 and by JSC "Gazprom Neft".

\end{document}